\documentclass{article}

\title{Estimating overidentified linear models with heteroskedasticity and outliers}
\author{Lei Bill Wang}
\date{}
\usepackage{setspace}
\usepackage[english]{babel}
\usepackage{amsmath}
\usepackage{amsthm}
\usepackage{mathtools}
\usepackage{blindtext}
\usepackage{geometry}
\usepackage{amsfonts}

\usepackage{graphicx}
\usepackage{caption}
\usepackage{subfig}

\usepackage{multirow}
\usepackage{indentfirst}

\usepackage{wrapfig}
\usepackage{lscape}
\usepackage{rotating}
\usepackage{epstopdf}

\usepackage[round]{natbib} 

\usepackage[dvipsnames]{xcolor}
 
\usepackage[T1]{fontenc}
\usepackage{array}
\usepackage{booktabs}
\usepackage{lscape}
\usepackage{algorithm}
\usepackage{algpseudocode}

\newcolumntype{L}[1]{>{\raggedright\arraybackslash}p{#1}}
\newcolumntype{C}[1]{>{\centering\arraybackslash}p{#1}}
\newcolumntype{R}[1]{>{\raggedleft\arraybackslash}p{#1}}
\providecommand{\keywords}[1]{\textbf{\textit{Keywords: }} #1}

\newtheorem{theorem}{Theorem}[section]
\newtheorem{lemma}{Lemma}[section]
\newtheorem{assumption}{Assumption}
\newtheorem*{assumption*}{Assumption}
\newtheorem{definition}{Definition}

\newtheorem{corollary}{Corollary}

\newtheorem{claim}{Claim}[section]

\usepackage[toc,page]{appendix}

\begin{document}
	\maketitle
\begin{abstract}
A large degree of overidentification causes severe bias in TSLS. A conventional heuristic rule used to motivate new estimators in this context is approximate bias. This paper formalizes the definition of approximate bias and expands the applicability of approximate bias to various classes of estimators that bridge OLS, TSLS, and Jackknife IV estimators (JIVEs). By evaluating their approximate biases, I propose new approximately unbiased estimators, including UOJIVE1 and UOJIVE2. UOJIVE1 can be interpreted as a generalization of an existing estimator UIJIVE1. Both UOJIVEs are proven to be consistent and asymptotically normal under a fixed number of instruments and controls. The asymptotic proofs for UOJIVE1 in this paper require the absence of high leverage points, whereas proofs for UOJIVE2 do not. In addition, UOJIVE2 is consistent under many-instrument asymptotic. The simulation results align with the theorems in this paper: (i) Both UOJIVEs perform well under many instrument scenarios with or without heteroskedasticity, (ii) When a high leverage point coincides with a high variance of the error term, an outlier is generated and the performance of UOJIVE1 is much poorer than that of UOJIVE2. 

\end{abstract}
\keywords{Approximate bias, instrumental variables (IV), overidentification, k-class estimator, Jackknife IV estimator (JIVE), many-instrument asymptotics, robustness to outlier}

\newpage
\section{Introduction}
Overidentified two-stage least squares (TSLS) are commonplace in economics research. \citet{mogstad2021causal} summarize that from January 2000 to October 2018, 57 papers from American Economic Review, Quarterly Journal of Economics, Journal of Political Economy, Econometrica, and the Review of Economic Studies adopt overidentified TSLS. 
Unfortunately, overidentification introduces a bias problem to TSLS. The intuition for the relationship between bias and the number of IVs can be illustrated with a pathological example in which a researcher adds so many instruments in the first-stage regression that the number of first-stage regressors (which include both the IVs and exogenous control variables)  is equal to the number of observations. The first stage regression equation has a perfect fit and its fitted value is exactly the observed endogenous variable values. Under this pathological example, TSLS and OLS perform the same, and so TSLS is equally biased as OLS.

To reduce or even completely resolve the problem, one may want to evaluate the bias of TSLS. However, evaluating the exact bias for estimators under the presence of endogeneity requires knowledge about the distribution of observable and unobservable variables. For example, \citet{anderson1982exact} and \citet{harding2016finite} assume that the error terms of the simultaneous equations are jointly normal and the assumption allows them to evaluate the finite sample distribution of TSLS and therefore, its bias. However, such assumptions can be too strong for economists' preferences. Therefore, many econometricians resort to a different evaluation criterion called ``approximate bias''. Many past works on IV estimators have evaluated the approximate bias of TSLS and used that to motivate new estimators \citep{nagar1959bias,fuller1977some,buse1992bias,angrist1999jackknife,ackerberg2009improved}. The idea behind approximate bias is to divide the difference between an estimator and the target parameter into two parts. One part is of a higher stochastic order than the other and is dropped out of the subsequent approximate bias calculation. The expectation of the lower stochastic order part is called approximate bias.

This paper formalizes the definition of approximate bias and shows why this definition is a reasonable heuristic rule for motivating new estimators. The definition of approximate bias proposed by \citet{nagar1959bias} applies only to the k-class estimator. On the other hand, the definition proposed in the working paper of \citet{angrist1999jackknife} (later referred to as AIK 1995) and used in \citet{ackerberg2009improved} (later referred to as AD 2009) applies to a large class of estimators which includes OLS, TSLS, all k-class estimators, JIVE1, IJIVE1 and UIJIVE1.\footnote{IJIVE1 and UIJIVE1 are originally termed IJIVE and UIJIVE, respectively, by \citet{ackerberg2009improved}. I attach the number 1 at the end of their names for comparison purposes which will become clear in later sections of this paper.} I show that the definition of approximate bias in AIK 1995 and AD 2009 is valid for an even larger class of estimators which additionally includes \citet{fuller1977some}, JIVE2, \citet{hausman2012instrumental}, and many other classes of estimators that bridge between OLS, TSLS, and JIVE2.

By expanding the applications of approximate bias to additional classes of estimators, I propose new estimators that are approximately unbiased. This paper focuses on two of the new estimators named UOJIVE1 and UOJIVE2. UOJIVE1 can be interpreted as a generalization of UIJIVE1. This paper shows that both UOJIVE1 and UOJIVE2 are consistent and asymptotically normal under a fixed number of instruments as the sample size goes to infinity. However, UOJIVE1's asymptotic proofs rely on an assumption that rules out high leverage points. On the other hand, UOJIVE2 does not require such an assumption. In addition, \citet{hausman2012instrumental}'s Theorem 1 directly applies to UOJIVE2, which proves that UOJIVE2 is consistent even under many-instrument asymptotic. Simulation results align with the theoretical results. UOJIVE1 and UOJIVE2 perform similarly well under a large degree of overidentification and heteroskedasticity. UOJIVE2's performance is much more stable than UOJIVE1's when an outlier (a high leverage point coinciding with a large variance of the error term) is deliberately introduced in the DGP.

The paper is organized as follows. Section \ref{Problem Setup} describes the problem setup and existing estimators in the approximate bias literature. Section \ref{sec def for approximate bias} defines approximate bias for a large class of estimators and explains the importance of this formal definition. Sections \ref{sec from UIJIVE1 to UIJIVE2} and \ref{sec from UIJIVE to UOJIVE} propose new estimators that are approximately unbiased. Section \ref{sec asymptotic proofs} shows that (i) under a fixed number of instruments, UOJIVE1 and UOJIVE2 are consistent and asymptotically normal, and (ii) under many-instrument asymptotic, UIJIVE2 is consistent. Section \ref{simulation} and section \ref{Empirical} demonstrate the superior performance of UOJIVE2 in simulation and empirical studies. Section \ref{conclusion} concludes the paper.

\section{Model setup and existing estimators}\label{Problem Setup}
\subsection{Overidentified linear model setup}
The model concerns a typical endogeneity problem: 
\begin{align}
	y_i & = X_i^* \beta^* + W_i\gamma^*+\epsilon_i \label{key1} \\
	X_i & = Z_i^*\pi^* + W_i\delta^* + \eta_i \label{key2}
\end{align}

Eq.(\ref{key1}) contains the parameter of interest $\beta^*$. $y_i$ is the response/outcome/dependent variable and $X_i^*$ is an $L_1$-dimensional row vector of endogenous covariate/explanatory/independent variables. Exogenous control variables, denoted as $W_i$, is an $L_2$-dimensional vector. Eq.(\ref{key2}) relates all the endogenous explanatory variables $X_i$ to instrumental variables, $Z_i$, and included exogenous control variables $W_i$ from Eq.(\ref{key1}). $Z_i$ is a $K_1$-dimensional row vector, where $K_1 \geq L_1$. Since this paper focuses on overidentified cases, I will assume that $K_1 > L_1$ throughout the rest of the paper. As \citet{hausman2012instrumental}, \citet{bekker2015jackknife}, this paper treats $Z$ as nonrandom (Alternatively, $Z$ can be assumed to be random, but conditioned on like in \citet{chao2012asymptotic}). $X_i$ is endogenous, ${\rm Cov}(\epsilon_i,X_i) = {\rm Cov}(\epsilon_i,Z_i\pi + \eta_i)= {\rm Cov}(\epsilon_i, \eta_i)  = \sigma_{\epsilon\eta} \neq 0$. 
I further assume that each pair of $(\epsilon_i,\eta_i)$ are independently and identically distributed with mean zero and covariance matrix
$\begin{pmatrix}
	\sigma_{\epsilon}^2 & \sigma_{\epsilon\eta} \\
	\sigma_{\eta\epsilon} & \sigma_{\eta}^2 
\end{pmatrix}$. $\sigma_{\epsilon}$ is a scalar. $\sigma_{\epsilon\eta}$ is a $L$-dimensional vector. $\sigma_{\eta}$ is a $L \times L$ matrix. 
I also impose a relevance constraint that $\pi \neq 0$. 
In matrix notation, I have Eqs.(\ref{key1}) and (\ref{key2}) as:
\begin{align}
	y = & X^*\beta^* + W\gamma^* + \epsilon \label{key3} \\
	X = & Z^*\pi^* + W\delta^* + \eta \label{key4}
\end{align}
where $y$ and $\epsilon$ are $(N \times 1)$ column vector; $X$ and $\eta$ are $(N \times L_1)$ matrices; $W$ is $(N \times L_2)$ and $Z$ is a $( N \times K_1)$ matrix, where $N$ is the number of observations. I also define the following notations for convenience:
\begin{equation*}
    \begin{aligned}
        X & = [X^*\quad W] \qquad \beta &&= \begin{pmatrix}
            \beta^* \\
            \gamma^*
        \end{pmatrix} \\
        Z & = [Z^*\quad W] \qquad \pi && = \begin{pmatrix}
            \pi^* & 0_{K_1\times L_2} \\
            \delta^* & I_{L_2}
        \end{pmatrix}
    \end{aligned}
\end{equation*}
Then, we have the following equivalent expressions for Eqs. (\ref{key3}) and (\ref{key4}):
\begin{align}
    y = X\beta + \epsilon \\
    X = Z\pi + \eta
\end{align}
It is also useful to define the following partialled out version of variables:
\begin{align*}
    \tilde{y} = y - W(W'W)^{-1}W'y \\
    \tilde{X} = X^* - W(W'W)^{-1}W'X^* \\
    \tilde{Z} = Z^* - W(W'W)^{-1}W'Z^* 
\end{align*}

\subsection{Existing estimators}
IV estimator is often used to solve this simultaneous equation problem. I tabulate in Table \ref{Table for a few estimators} some of the existing IV estimators this paper repeatedly refers to. They all have the matrix expression $(X'C'X)^{-1}(X'C'y)$. The caption of Table \ref{Table for a few estimators} summarizes how these estimators are linked to each other. 

\begin{table}
	\centering
	\renewcommand{\arraystretch}{1.25}
	\begin{tabular}{| c | c |}
		\hline
		Estimators& C \\
            \hline
		OLS &   $I$ \\
		TSLS &    $P_Z$ \\
            k-class & $(1-k)I + kP_Z$ \\
            JIVE2 &    $P_Z-D$ \\
		JIVE1 &    $(I-D)^{-1}(P_Z-D)$ \\
            IJIVE1 & $(I-\tilde{D})^{-1}(P_{\tilde{Z}}-\tilde{D})$\\
            UIJIVE1 & $(I-\tilde{D}+ \omega I)^{-1}(P_{\tilde{Z}}-\tilde{D} + \omega I)$ \\
		\hline
	\end{tabular}
	\caption{All these IV estimators are of the analytical form $(X'C'X)^{-1}(X'C'y)$. $D$ is the diagonal matrix of the projection matrix $P_Z = Z(Z'Z)^{-1}Z'$. $\tilde{Z}$ is $Z$ partialled out by $W$, $\tilde{Z} = Z^* - W(W'W)^{-1}W'Z^*$. $P_{\tilde{Z}}$ is the projection matrix of $\tilde{Z}$ and $\tilde{D}$ is the diagonal matrix of $P_{\tilde{Z}}$. JIVE2 modifies TSLS by removing the diagonal entries of the projection matrix $P_Z$. JIVE1 adds a rowwise division operation in front of the $C$ matrix of JIVE2. IJIVE is essentially equivalent to JIVE1, the only difference is that IJIVE takes in $(\tilde{y},\tilde{X},\tilde{Z})$ instead of $(y,X,Z)$. Its closed-form is written as $(\tilde{X}(I-\tilde{D})^{-1}(P_{\tilde{Z}}-\tilde{D})\tilde{X})^{-1}\tilde{X}(I-\tilde{D})^{-1}(P_{\tilde{Z}}-\tilde{D})\tilde{y})$. IJIVE reduces the approximate bias of JIVE1. UIJIVE1 further reduces the approximate bias of IJIVE1 by adding a constant $\omega$ at the diagonal of the inverse term and the term post-multiplied to the inverse in the $C$ matrix. $\omega = \frac{L_1+1}{N}$.}
    \label{Table for a few estimators}
\end{table}

\section{Approximate bias}\label{sec def for approximate bias}
IV estimator is often employed to estimate $\beta^*$ (or $\beta$) in Section \ref{Problem Setup}. The most commonly used IV estimator is TSLS which has a bias problem when the degree of overidentification is large. Unfortunately, completely removing the bias of overidentified TSLS is generally infeasible unless economists are willing to assume parametric families for instrumental variables, $Z$. Therefore, econometricians often resort to a concept called \textit{approximate bias}. For example, JIVE1, JIVE2, IJIVE1, and UIJIVE1 from Table \ref{Table for a few estimators} are all motivated by reducing approximate bias. The intuition behind the idea is to divide the difference between an estimator, $\hat{\beta}$, and the true parameter, $\beta$, that the estimator is aiming to estimate into two parts. One part is of a higher stochastic order than the other and therefore, is dropped out of the subsequent approximate bias calculation. The other part with lower stochastic order has an easy-to-evaluate expectation. Its expectation is called approximate bias. 

\subsection{Definition of approximate bias}
The following definition of approximate bias has been used by AIK 1995 and AD 2009 to motivate their development of new estimators whose $C$ matrix satisfies property $CZ = Z$ (and hence $CX = Z\pi + C\eta$):  
\begin{definition} \label{def for approximate bias}
    The approximate bias of an IV estimator is $E[R_N]$ where 
	\begin{equation*}
		R_N = J\epsilon - \frac{Q_0}{N}\pi'Z'\eta J \epsilon +\frac{Q_0}{N}\eta' C'\epsilon - \frac{Q_0}{N}\eta'P_{Z\pi}\epsilon 
	\end{equation*}
in which $Q_0 = \lim_{N\to \infty} (\pi' \frac{Z'Z}{N} \pi)^{-1}$, $J = (\pi'Z'Z\pi)^{-1}\pi'Z'$ and $P_{Z\pi} = Z\pi(\pi'Z'Z\pi)^{-1}\pi'Z'$. 
\end{definition}
For readers' convenience, I justify why this definition can be used as a reasonable heuristic rule for motivating new estimators in Appendix \ref{def for approximate bias derivation}. Note that both AIK 1995 and AD 2009 assume that $CZ = Z$ which restricts the application of Definition \ref{def for approximate bias}. The condition that $CZ=Z$ is satisfied by OLS, TSLS, all k-class estimators, JIVE1, IJIVE1, and UIJIVE1, but not many other estimators such as JIVE2, HLIM, and HFUL from \citet{hausman2012instrumental}, $\lambda_2$- and $\omega_2$-class estimators which I will introduce in later sections of the paper. In Appendix \ref{Approximate bias for other classes of estimators}, I show that the condition $CZ = Z$ is not necessary, Definition \ref{def for approximate bias} is a reasonable heuristic rule for the estimators as mentioned earlier which do not satisfy the condition $CZ = Z$.

\subsection{Approximate bias of existing estimators} \label{sec Approximate bias of existing estimators}

Definition \ref{def for approximate bias} yields Corollary \ref{corollary for approximate bias} and Definitions \ref{def approximately unbiased} and \ref{def approximate bias asymptotic vanishing}. I will use them to analyze the existing IV estimators in Table \ref{Table for a few estimators}.

\begin{corollary} \label{corollary for approximate bias}
	Approximate bias of an IV estimator is $(X'C'X)^{-1}(X'C'y)$ is 
	\[
	\frac{Q_0}{N}(tr(C')-tr(P_{Z\pi})-1)\sigma_{\eta\epsilon}
	\]
	where $Q_0 = \lim_{N\to \infty} \pi' \frac{Z'Z}{N} \pi$ and $P_{Z\pi} = Z\pi(\pi'Z'Z\pi)^{-1}\pi'Z'$. An estimator's approximate bias is proportional to $tr(C) - \mathcal{L} - 1$, where $\mathcal{L}$ is the number of columns of $Z\pi$ (or $X$).\footnote{I use $\mathcal{L}$ instead of $L$ because while for most of the estimators in Table \ref{Table for a few estimators}, $\mathcal{L} = L$; for IJIVE1 and UIJIVE1, $\mathcal{L} = L_1$.}
\end{corollary}

\begin{definition} \label{def approximately unbiased}
    An estimator is said to be approximately unbiased if $tr(C)-tr(P_{Z\pi})-1 = 0$.
\end{definition}

\begin{definition} \label{def approximate bias asymptotic vanishing}
The approximate bias of an estimator is said to be asymptotically vanishing if 
 \begin{equation*}
     tr(C)-tr(P_{Z\pi})-1 \overset{}{\to} 0 \text{ as } N \to \infty.
 \end{equation*} 
\end{definition}


I compute the approximate bias by showing the term $tr(C) - \mathcal{L} - 1$ for the following estimators. Larger magnitude of $tr(C) - \mathcal{L} - 1$ means a larger approximate bias. Ideally, $tr(C) - \mathcal{L} - 1 = 0$, so that we can claim the estimator to be approximately unbiased. 

The approximate bias computation for OLS, TSLS, and k-class estimators is straightforward: \textbf{OLS}'s approximate bias is proportional to $N - L - 1$, \textbf{TSLS}'s approximate bias is proportional to $K - L - 1$, \textbf{k-class estimator}'s approximate bias is proportional to $kK + (1-k)N - L - 1$, where $k$ is a user-defined parameter. Setting $k = \frac{N - L - 1}{N - K}$ gives us an approximately unbiased estimator. I call this estimator AUK (Approximately unbiased k-class estimator). 

As explained earlier, TSLS' approximate bias is proportional to $K - L - 1$ which is large when the degree of overidentification is large. AIK 1995 proposes a jackknife version of TSLS to reduce the approximate bias of TSLS when the degree of overidentification is large. The authors call the proposed estimators JIVE1 and JIVE2 (Jackknife IV Estimator). Evaluating the approximate bias of \textbf{JIVE2} is straight forward:
\begin{equation*}
    tr(P_Z - D) - \mathcal{L} - 1 =  -L -1 
\end{equation*}

We obtain the approximate bias of \textbf{JIVE1}:
\begin{equation*}
    tr((I - D)^{-1}(P_Z - D)) - \mathcal{L} - 1 = \sum_{i=1}^N \frac{0}{1-D_{i}} -L -1 = -L-1
\end{equation*}

AD 2009 reduces the approximate bias of the two JIVEs by paritialling out $W$ from $(y, X^*, Z^*)$. We have the approximate bias of IJIVE1 proportional to
\begin{equation*}
    tr((I-\tilde{D})^{-1}(P_{\tilde{Z}}-\tilde{D})) - \mathcal{L} - 1 = \sum_{i=1}^N \frac{0}{1-\tilde{D}_{i}} -L_1-1 = -L_1 -1
\end{equation*}
Note that $L_1 < L$, so IJIVE1 potentially has a smaller approximate bias than JIVEs.\footnote{The comparison requires knowledge on $\lim_{N\to \infty} \pi' \frac{\tilde{Z}'\tilde{Z}}{N} \pi$. However, since approximate bias is used only as a heuristic rule to motivate new estimators, evaluating $tr(C)-tr(P_{Z\pi})-1$ is arguably sufficient for motivating purposes.} AD 2009 also further reduces approximate bias by adding a constant ($\frac{L_1+1}{N}$) to the diagonal of the $C$ matrix of IJIVE1. The new estimator is called UIJIVE1. 

To evaluate the approximate bias of \textbf{UIJIVE1}, I make the following assumption about the leverage of $Z$, ${\{{D}_i\}}_{i=1}^N$:
\begin{assumption*}[\textbf{BA}] 
    $\max_i {D}_i$ is \textbf{bounded away} from 1 for large enough N from 1.
    Equivalently, $\exists~m > 0, \text{ such that }  {D}_i \leq 1-m$  for large enough $N$ and for all $i = 1, 2, 3, \dots, N$.
\end{assumption*}
I state the assumption in terms of $Z$ as I will repeatedly use it for $Z$ in Section \ref{sec asymptotic proofs}. Here, to compute the approximate bias of \textbf{IJIVE1}, I make assumption BA, but for ${\{\tilde{D}_i\}}_{i=1}^N$ instead of for ${\{{D}_i\}}_{i=1}^N$. Under Assumption BA for ${\{\tilde{D}_i\}}_{i=1}^N$, \textbf{UIJIVE1}'s approximate bias is proportional to 
\begin{equation*}
    tr(C) - \mathcal{L} - 1 = \sum_{i=1}^{N} \frac{\frac{L_1+1}{N}}{1 - \tilde{D}_i + \frac{L_1+1}{N}} - L_1 - 1 \to 0
\end{equation*}
Therefore, we have that UIJIVE1's approximate bias is asymptotically vanishing. See proof in Appendix \ref{appendix UIJIVE1 approximate bias asymptotically vanishing}.

\section{From UIJIVE1 to UIJIVE2}\label{sec from UIJIVE1 to UIJIVE2}

This section interprets the relationship between existing estimators (in particular, JIVE1, IJIVE1, UIJIVE1, and OLS) which sheds light on how a new estimator that is approximately unbiased, namely, UIJIVE2 is developed. 

\subsection{$\omega_1$-class estimator and UIJIVE1}
I define the $\omega_1$-class estimator that contains all the following four estimators: JIVE1, IJIVE1, UIJIVE1, and OLS. Since all of them can be expressed as: $(X'C'X)^{-1}(X'C'y)$, I define matrix $C$ for the $\omega_1$-class estimators:
\begin{equation*}
    (I - D + \omega_1 I)^{-1}(P_Z - D + \omega_1 I)
\end{equation*}
when $\omega_1 = 0$, it corresponds to JIVE1; when $\omega_1 = \infty$, it corresponds to OLS. On the other hand, IJIVE1 is a special case of JIVE1 where $(y, X, Z)$ is replaced with $(\tilde{y},\tilde{X},\tilde{Z})$, and hence is belonged to $\omega_1$-class estimator. UIJIVE1 uses $(\tilde{y},\tilde{X},\tilde{Z})$ and sets $ \omega_1 = \frac{L_1+1}{N}$. As stated in Section \ref{sec Approximate bias of existing estimators}, this choice of $\omega_1 = \frac{L_1+1}{N}$ outputs UIJIVE1 whose approximate bias is asymptotically vanishing. The information of these four estimators is summarized in Table \ref{omega1 a few estimators}. 

The development from TSLS to JIVE1 to IJIVE1 to UIJIVE1 is depicted in the upper half of Figure \ref{Figure Estimator development history}. 
\textbf{TSLS to JIVE1}: JIVE1 is a jackknife version of TSLS where the first stage OLS in TSLS is replaced with a jackknife procedure for out of sample prediction. Call the resulting matrix from the first stage $\hat{X}$, the second stage IV estimation is the same for TSLS and JIVE1: $(\hat{X}'X)^{-1}(\hat{X}'y)$.
\textbf{JIVE1 to IJIVE1}: IJIVE1 is a special case of JIVE1. JIVE1 takes $(y,X,Z)$ as the input, IJIVE1 takes $(\tilde{y},\tilde{X},\tilde{Z})$ (the partialled-out version of $(y,X,Z)$) as input. \textbf{IJIVE1 to UIJIVE1}: UIJIVE1 adds a constant $\omega_1 = \frac{L_1+1}{N}$ to the diagonal of $C$ matrix of IJIVE1. We can interpret this addition of constant $\omega_1 I$ as bridging IJIVE1 and OLS since OLS has its $C$ matrix as $I$.

\begin{table}
    \centering
    \subfloat[OLS can take either $(y,X,Z)$ or $(\tilde{y},\tilde{X},\tilde{Z})$. It estimates the entire $\beta$ when it takes $(y, X, Z)$, it estimates the parameters for only the endogenous variables $\beta^*$ when it takes $(\tilde{y},\tilde{X},\tilde{Z})$. UIJIVE1 can be interpreted as an estimator that bridges IJIVE and OLS.]{\begin{tabular}{| c | c c |}
        \hline
         & $(y,X,Z)$ or $(\tilde{y},\tilde{X},\tilde{Z})$  & $\omega_1$\\
         \hline
         JIVE1& $(y,X,Z)$ & 0 \\
         IJIVE& $(\tilde{y},\tilde{X},\tilde{Z})$ & 0\\
         UIJIVE1& $(\tilde{y},\tilde{X},\tilde{Z})$ & $\frac{L_1+1}{N}$ \\
         OLS& Both & $\infty$\\
         \hline
    \end{tabular}\label{omega1 a few estimators}}
    \quad
    \subfloat[OLS can take either $(y,X,Z)$ or $(\tilde{y},\tilde{X},\tilde{Z})$. It estimates the entire $\beta$ when it takes $(y,X,Z)$, it estimates the parameters for only the endogenous variables $\beta^*$ when it takes $(\tilde{y},\tilde{X},\tilde{Z})$. UIJIVE2 can be interpreted a an estimator that bridges IJIVE2 and OLS.]{\begin{tabular}{| c | c c |}
        \hline
         & $(y,X,Z)$ or $(\tilde{y},\tilde{X},\tilde{Z})$  & $\omega_2$\\
         \hline
         JIVE2& $(y,X,Z)$ & 0 \\
         IJIVE2 & $(\tilde{y},\tilde{X},\tilde{Z})$ & 0\\
         UIJIVE2& $(\tilde{y},\tilde{X},\tilde{Z})$ & $\frac{L_1+1}{N}$ \\
         OLS& Both & $\infty$\\
         \hline
    \end{tabular}\label{omega2 a few estimators}}

    \caption{Some examples of $\omega_1$-class and $\omega_2$-class estimators. The left panel are from $\omega_1$-class and the right panel are from $\omega_2$-class.}
    \label{fig:enter-label}
\end{table}

\subsection{$\omega_2$-class estimator and UIJIVE2}
\begin{figure}
    \centering
    \includegraphics[width = .8\textwidth]{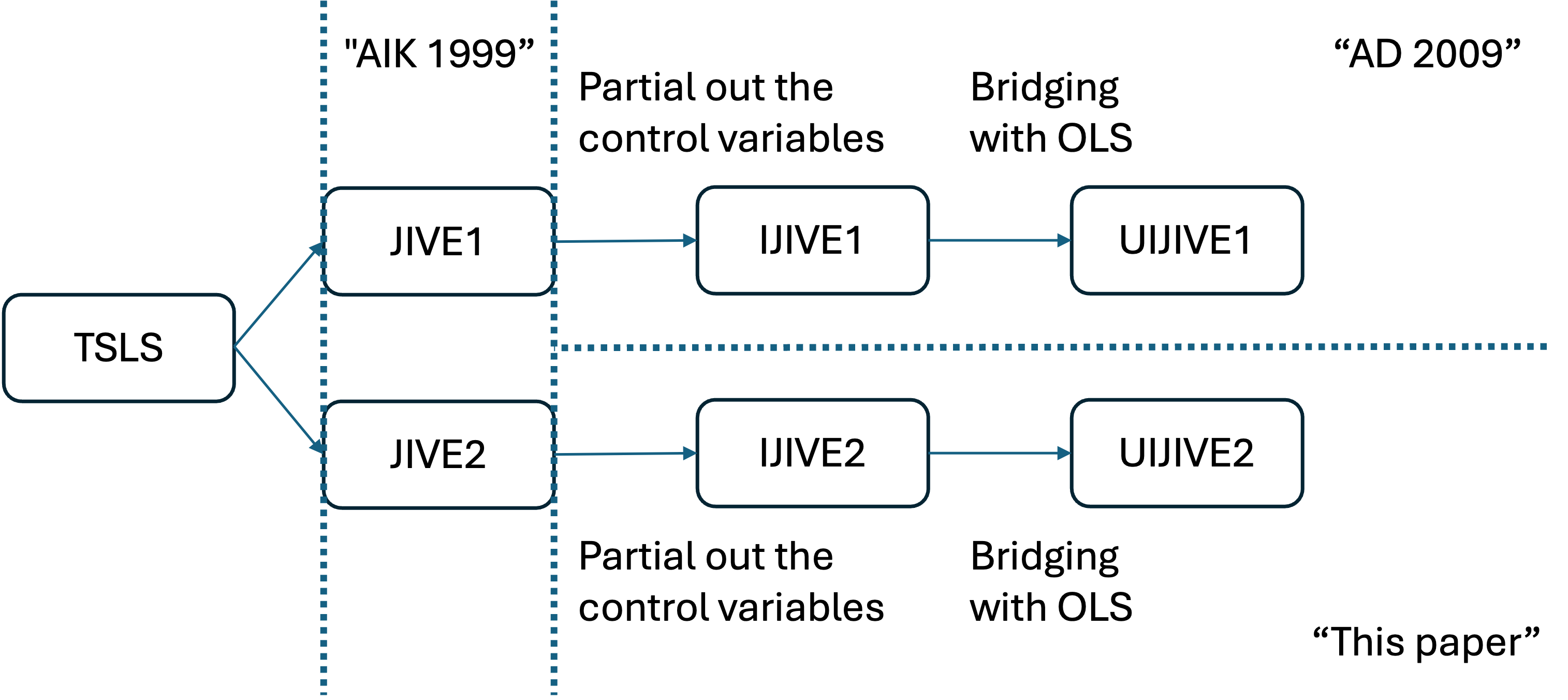}
    \caption{The development of past estimators and new estimators. The papers that develop the corresponding estimators are in quotation marks. AIK 1999 develop JIVE1 and JIVE2 and AD 2009 develop IJIVE and UIJIVE1. This paper develops new estimators IJIVE2 and UIJIVE2. As shown in the figure, the mathematical relationship between {JIVE1, IJIVE1, and UIJIVE1} is similar to that between {JIVE2, IJIVE2, and UIJIVE2}. The relationship between UIJIVE1 and UIJIVE2 (or between IJIVE1 and IJIVE2, which is not of interest to this paper) is analogous to the relationship between {JIVE1 and JIVE2}. Latter estimators (UIJIVE2 and JIVE2) remove the row-wise division in the former estimators (UIJIVE1 and JIVE1).}
    \label{Figure Estimator development history}
\end{figure}

I define $\omega_2$-class of estimator that contains JIVE2 and OLS. The $C$ matrix of $\omega_2$-class estimator is 
\begin{equation*}
    P_Z - D + \omega_2 I.
\end{equation*}
The only difference between $C$ matrices for $\omega_1$- and $\omega_2$- class estimators is the exclusion/inclusion of $(I - D + \omega_1 I)^{-1}$, which is a rowwise division operation (i.e. dividing $i$th row of $P_Z - D + \omega_1 I$ by $1 - D_i + \omega_1$).

$\omega_2 = 0$ corresponds to JIVE2 and $\omega_2 = \infty$ corresponds to OLS. By Corollary \ref{corollary for approximate bias}, the approximate bias of $\omega_2$-class estimators that takes $(\tilde{y}, \tilde{X}, \tilde{Z})$ is proportional to $\omega_2 N - L_1 - 1$. If $\omega_2 = 0$, I obtain a new estimator IJIVE2, corresponding to IJIVE1 from $\omega_1$-class estimators. Selecting $\omega_2 = \frac{L_1+1}{N}$, I obtain approximately unbiased $\omega_2$-class estimator and name it UIJIVE2. Its closed-form expression is:
\begin{equation*}
    (\tilde{X}'(P_{\tilde{Z}} - \tilde{D} + \frac{L_1+1}{N} I)'\tilde{X})^{-1}(\tilde{X}'(P_{\tilde{Z}} - \tilde{D} + \frac{L_1+1}{N} I)'\tilde{y})
\end{equation*}
where the transpose on $P_{\tilde{Z}} - \tilde{D} + \frac{L_1+1}{N} I$ is not necessary as the matrix is symmetric, but I keep it for coherent notation. 

The information of OLS, JIVE2, IJIVE2, and UIJIVE2 are summarized in Table \ref{omega2 a few estimators}. I also depict the parallel relationship between $\omega_1$-class and $\omega_2$-class estimators in the lower half of Figure \ref{Figure Estimator development history}. The development from OLS to JIVE2 to IJIVE2 to UIJIVE2 is analogous to the development of the $\omega_1$-class estimators. Readers can also interpret Figure \ref{Figure Estimator development history} from top to bottom, all the estimators in the bottom remove the rowwise division operations in the estimators above them. The removal of rowwise division is important to the comparison between $\omega_1$- and $\omega_2$-class estimators. The division operation is highly unstable when $D_i$ (or $\tilde{D}_i$) is close to 1, leading undesirable asymptotic property to $\omega_1$-class estimator, but not to $\omega_2$-class estimator. This point will be explained formally in Section \ref{sec asymptotic proofs}.

\begin{table}
 \centering
 
		\renewcommand{\arraystretch}{1.75}
		\begin{tabular}{c c  c  c c }
			\hline
			Name & Consistency & Approximately unbiased& \multicolumn{2}{c}{Many-instrument consistency} \\ 
			&&&Homoskedasticity&Heteroskedasticity \\
			\hline\hline 
			OLS&  & & & \\
			TSLS& \checkmark& & &  \\
			Nagar&\checkmark  & \checkmark$^*$ & &  \\
			JIVE1& \checkmark& & \checkmark& \\
			JIVE2& \checkmark& & \checkmark& \\
			UIJIVE1& \checkmark& \checkmark& \checkmark& \checkmark\\
			UIJIVE2& \checkmark& \checkmark& \checkmark& \checkmark\\
			\hline\hline
		\end{tabular}
		\caption{Properties of different estimators in the approximate bias literature with endogenous regressor. $^*$ means that Nagar estimator's approximately unbiased property is only true under homoskedasticity, but not under heteroskedasticity, see proof in AD 2009.} 
		\label{table estimator summary}
	\end{table}

\section{From UIJIVE to UOJIVE} \label{sec from UIJIVE to UOJIVE}
Figure \ref{Figure Estimator development history} shows that UIJIVEs can be interpreted as an approximately unbiased estimator selected from a class of estimators that bridges IJIVEs and OLS. I apply the same thought process to other classes of estimators that bridge between OLS, TSLS, and JIVEs to obtain new approximately unbiased estimators. Table \ref{Table for a few estimators in extension} summarizes these five classes of estimators (k-class, $\omega_1$-, $\omega_2$-, $\lambda_1$-, and $\lambda_2$-classes) and approximately unbiased estimators. See Appendix \ref{sec other classes of estimators} for details on these estimators. The relationships between five classes of estimators are illustrated in Figure \ref{three classes}.

\begin{table}
	\centering
	\renewcommand{\arraystretch}{2}
	\begin{tabular}{| c | c | c | c |}
		\hline
		Class & Estimators& C & parameter\\
            \hline
		$k$-class & AUK &  $kP_Z +(1-k)I$ & $\frac{N-L-1}{N-K}$\\
		$\lambda_1$-class&TSJI1 &    $(I - \lambda_1 D)(P_Z - \lambda_1 D)$ & $\frac{K-L-1}{K}$ \\
            $\lambda_2$-class&TSJI2 &    $P_Z- \lambda_2 D$ & $\frac{K-L-1}{K}$\\
		$\omega_1$-class & UOJIVE1 &    $(I-D+ \omega_1 I)^{-1}(P_Z-D+ \omega_1 I)$ & $\frac{L+1}{N}$\\
            $\omega_2$-class& UOJIVE2 & $P_Z - D + \omega_2 I$& $\frac{L+1}{N}$ \\
		\hline
	\end{tabular}
	\caption{$D$ is the diagonal matrix of the projection matrix $P_Z = Z(Z'Z)^{-1}Z'$. ``parameter'' column states the parameter value such that the estimator is approximately unbiased.}
	\label{Table for a few estimators in extension}
\end{table}

Note that UIJIVE1 and UIJIVE2 can be considered as special cases of UOJIVE1 and UOJIVE2. UIJIVE can be understood as the following two steps:
\begin{enumerate}
    \item Partial out $W$ from $Z^*$,$X^*$ and $y$. $\tilde{Z} = Z^* - P_WZ^*$. $\tilde{X} = X^* - P_WX^*$. $\tilde{y} = y - P_Wy$.
    \item Set $\omega = \frac{L_1+1}{N}$ and compute the estimate as an $\omega_1$- and $\omega_2$-class estimator using $\tilde{Z}$, $\tilde{X}$ and $\tilde{y}$.
\end{enumerate}
The second step is exactly the UOJIVE for the following model:
\begin{align*}
	\tilde{y} =& \tilde{X}\beta + \tilde{\epsilon} \\
	\tilde{X} = & \tilde{Z}\pi + \tilde{\eta}
\end{align*}
where $\tilde{\epsilon} = \epsilon - P_W \epsilon$ and $\tilde{\eta} = \eta - P_W \eta$.

For the rest of the paper, I will show the asymptotic properties of UOJIVE1 and UOJIVE2 and demonstrate that UOJIVE2 is more robust to outliers than UOJIVE1. The theoretical results can be easily generalized to UIJIVE1, UIJIVE2, TSJI1 and TSJI2. 


\begin{figure}

    \centering
    \subfloat[k, $\lambda_1$- and $\omega_1$-classes estimators]{\includegraphics[width=0.45\textwidth]{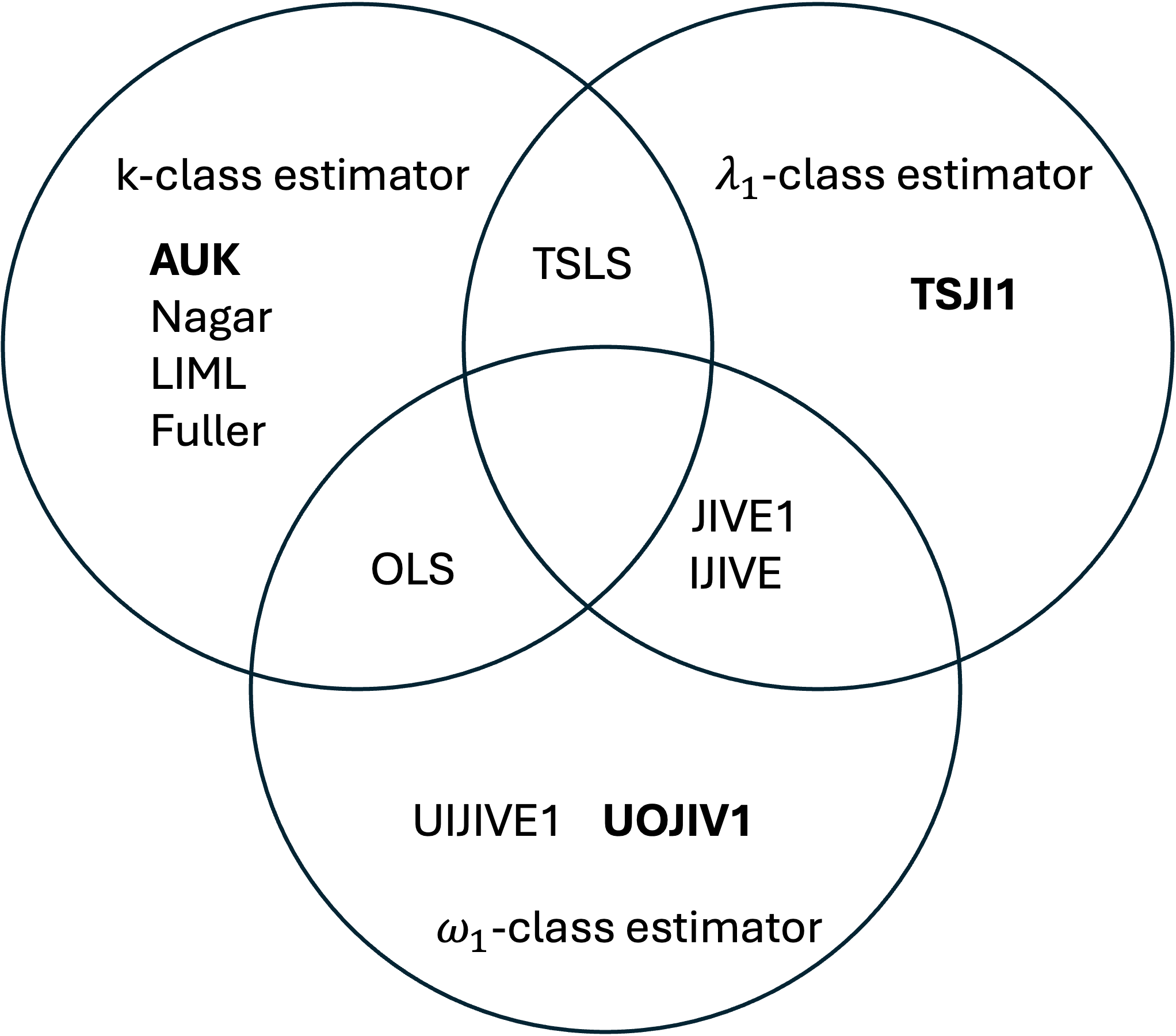}\label{three classes 1}}\hfill
    \subfloat[k, $\lambda_2$- and $\omega_2$-class estimators]{\includegraphics[width=0.45\textwidth]{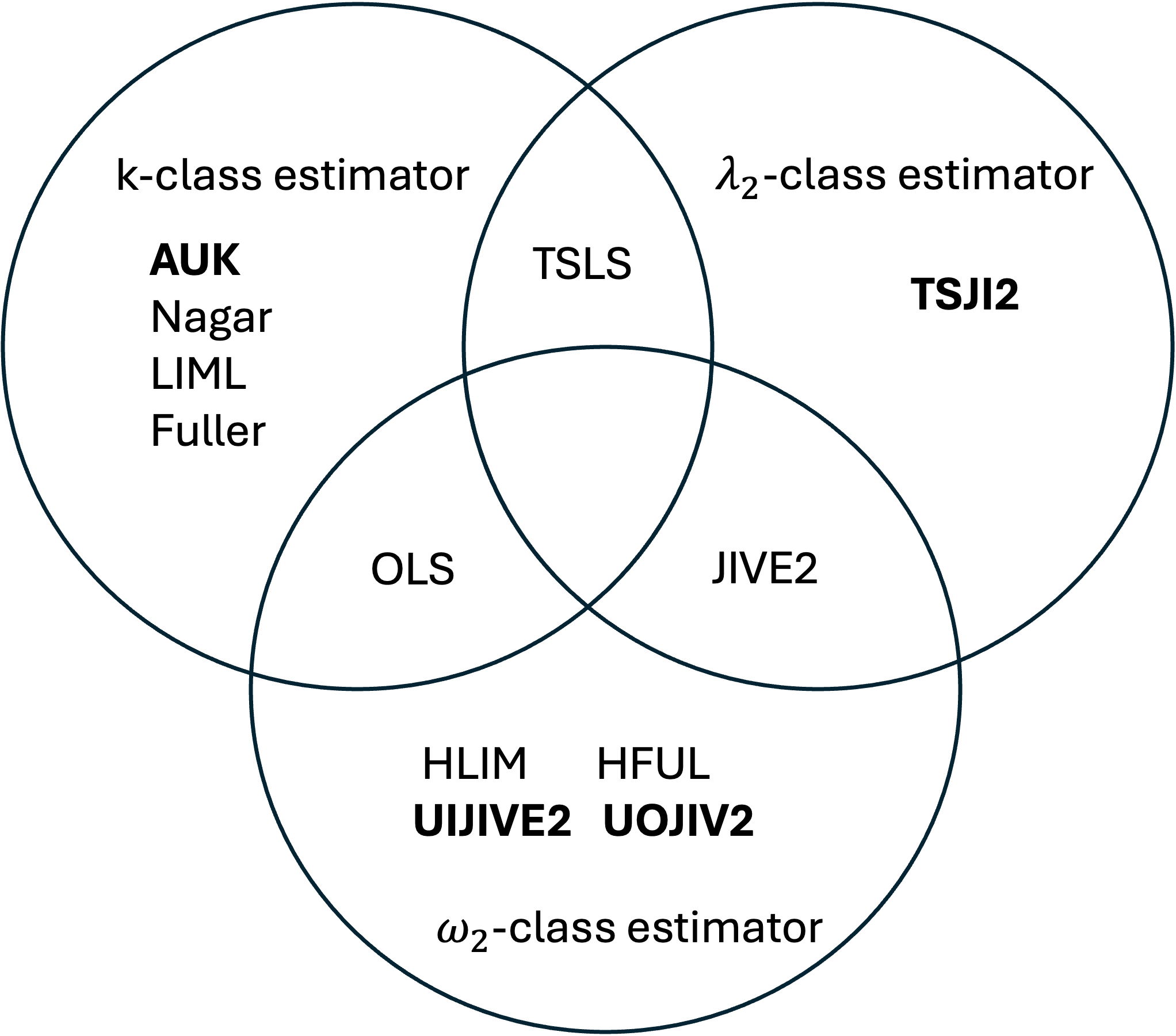}\label{three classes 2}}

    \caption{The two figures illustrate the connections between various classes of estimators. The estimators in bold are those proposed by this paper. The relationship between TSJI1 and TSJI2 is analogous to the relationship between JIVE1 and JIVE2. The same analogy applies to the relationship between UOJIVE1 and UOJIVE2. The left panel contains three classes of estimators whose $C$ matrices have the $CZ=Z$ property whereas $\lambda_2$- and $\omega_2$-class estimators on the right panel do not have this property.}
	\label{three classes}
\end{figure}

\section{Asymptotic property of UOJIVE1 and UOJIVE2} \label{sec asymptotic proofs}

\subsection{Consistency of UOJIVEs under fixed $K$}
Under fixed $K$ and $L$, I show that both UOJIVE1 and UOJIVE2 are consistent as $N \to \infty$ with some assumptions imposed on the moment existence for observable and unobservable variables. It is worth mentioning that Assumption BA is important to the consistency of UOJIVE1, but not to that of UOJIVE2, suggesting that UOJIVE2 is robust to the presence of high leverage points while UOJIVE1 is not. Throughout this subsection, I make the following assumptions

\begin{assumption}\label{regularity}
	Standard regularity assumptions hold for 
	\begin{align*}
		\frac{1}{N} X'Z \overset{p}{\to}& \Sigma_{XZ},\\
		\frac{1}{N} Z'X \overset{p}{\to}& \Sigma_{ZX},\\
		\frac{1}{N} Z'Z \overset{}{\to}& \Sigma_{ZZ},\\
		\frac{1}{\sqrt{N}} Z'\epsilon \overset{d}{\to}& N(0,\sigma_{\epsilon}^2\Sigma_{ZZ}).
	\end{align*}
\end{assumption}

\begin{assumption} \label{finite XX}
    $E[\lVert Z_i'X_i\rVert ^{1+\delta_1}]$ and $E[\lVert \eta_i'X_i\rVert ^{1+\delta_1}]$ are finite for some $\delta_1 > 0$. The existence of these two expectations jointly implies that $E[\lVert X_i'X_i\rVert^{1+\delta_1}]$ is finite. 
\end{assumption}

\begin{assumption}\label{finite Xepsilon lower}
	$E[\lVert Z_i'\epsilon_i\rVert^{1+\delta_1}]$ and $E[\lVert \eta_i'\epsilon_i\rVert^{1+\delta_1}]$ are finite, for some $\delta_1 > 0$. The existence of these two expectations jointly implies that $E[\lVert X_i'\epsilon_i\rVert^{1+\delta_1}]$ is finite.
\end{assumption}

Denote $\frac{L+1}{N}$ as $\omega$. Recall that UOJIVE1's matrix expression is
\begin{align*}
    \hat{\beta}_{UOJIVE1} = &(X'(P_Z - D +\omega I)(I -  D + \omega I)^{-1}X)^{-1}(X'(P_Z - D +\omega I)(I - D + \omega I)^{-1}y)\\
	=& \beta + (X'(P_Z - D +\omega I)(I -  D + \omega I)^{-1}X)^{-1}(X'(P_Z - D +\omega I)(I -  D + \omega I)^{-1}\epsilon)
\end{align*}

\begin{lemma} \label{prop UOJIVE1 consistency A}
    Under Assumption BA, \ref{regularity}, and \ref{finite XX}, $\frac{1}{N}X'(P_Z - D +\omega I)(I -  D + \omega I)^{-1}X \overset{p}{\to} H$, where $H = \Sigma_{XZ}\Sigma_{ZZ}^{-1}\Sigma_{ZX}$.
\end{lemma}

\begin{lemma} \label{prop UOJIVE1 consistency B}
    Under Assumption BA, \ref{regularity}, and \ref{finite Xepsilon lower}, $\frac{1}{N}X'(P_Z - D +\omega I)(I -  D + \omega I)^{-1}\epsilon \overset{p}{\to} 0$.
\end{lemma}

The proofs for Lemma~\ref{prop UOJIVE1 consistency A} and \ref{prop UOJIVE1 consistency B} are collected in Appendix \ref{appendix asymptotic proofs}. The importance of Assumption BA stems from the presence of $\{D_i\}_{i=1}^N$ in the denominators. UOJIVE2 does not have the same problem since $(I -  D + \omega I)^{-1}$ does not appear in its analytical form.

\begin{theorem}
    Under Assumption BA, \ref{regularity}, \ref{finite XX}, \ref{finite Xepsilon lower}, $\hat{\beta}_{UOJIVE1} \overset{p}{\to} \beta$.
\end{theorem}
\begin{proof}
    With Lemma~\ref{prop UOJIVE1 consistency A} and \ref{prop UOJIVE1 consistency B}, proving the theorem is trivial.
\end{proof}

Now, we turn to analyzing the consistency of UOJIVE2. The steps for proving its consistency are similar to what we have done for UOJIVE1.

\begin{lemma} \label{prop UOJIVE2 consistency A}
    Under Assumption \ref{regularity} and \ref{finite XX}, $\frac{1}{N}X'(P_Z - D +\omega I)X \overset{p}{\to} H$, where $H = \Sigma_{XZ}\Sigma_{ZZ}^{-1}\Sigma_{ZX}$.
\end{lemma}

\begin{lemma} \label{prop UOJIVE2 consistency B}
    Under Assumption \ref{regularity} and \ref{finite Xepsilon lower}, $\frac{1}{N}X'(P_Z - D +\omega I)\epsilon \overset{p}{\to} 0$.
\end{lemma}

The proofs for Lemma~\ref{prop UOJIVE2 consistency A} and \ref{prop UOJIVE2 consistency B} are collected in Appendix \ref{appendix asymptotic proofs}. With these two lemmas, we establish the consistency of UOJIVE2.
\begin{theorem} \label{consistency of UIJIVE2}
    Under Assumption \ref{regularity}, \ref{finite XX}, and \ref{finite Xepsilon lower}, $\hat{\beta}_{UOJIVE2} \overset{p}{\to} \beta$.
\end{theorem}

\subsection{Asymptotic variance of UOJIVEs under fixed $K$}
\begin{assumption}\label{finite Xepsilon}
	$E[\lVert Z_i'\epsilon_i\rVert^{2+\delta_1}]$ and $E[\lVert \eta_i'\epsilon_i\rVert^{2+\delta_1}]$ are finite, for some $\delta_1 > 0$. The existence of these two expectations jointly implies that $E[\lVert X_i'\epsilon_i\rVert^{2+\delta_1}]$ is finite.
\end{assumption}

\begin{lemma} \label{prop asymptotic variance UOJIVE1 B}
    Under Assumption BA, \ref{regularity}, \ref{finite Xepsilon}, $\frac{1}{\sqrt{N}}X'(P_Z - D +\omega I)(I -  D + \omega I)^{-1}\epsilon \overset{d}{\to} N(0,\sigma_\epsilon^2 H)$.
\end{lemma}
Proof for Lemma~\ref{prop asymptotic variance UOJIVE1 B} can be found in Appendix \ref{appendix asymptotic proofs}. $\frac{1}{\sqrt{N}}X'(P_Z - D +\omega I)(I -  D + \omega I)^{-1}\epsilon \overset{d}{\to} N(0,\sigma_\epsilon^2H)$. Combining with Lemma \ref{prop UOJIVE1 consistency A}, the asymptotic variance of UOJIVE is $\sigma_\epsilon^2H^{-1}$.

\begin{theorem}
    Under Assumption BA, \ref{regularity}, \ref{finite XX}, \ref{finite Xepsilon}, $\sqrt{N}(\hat{\beta}_{UOJIVE1} - \beta) \overset{d}{\to} N(0, \sigma_\epsilon^2 H^{-1})$.
\end{theorem}

\begin{lemma} \label{prop asymptotic variance UOJIVE2 B}
    Under Assumption \ref{regularity} and \ref{finite Xepsilon}, $\frac{1}{\sqrt{N}}{X}'(P_{{Z}} - {D} +\omega I)'{\epsilon} \overset{d}{\to} N(0,\sigma_{{\epsilon}}^2 H)$.
\end{lemma}
Proof for the lemma can be found in Appendix \ref{appendix asymptotic proofs}. Lemma \ref{prop UOJIVE2 consistency A} and \ref{prop asymptotic variance UOJIVE2 B} establish the asymptotic normality of UOJIVE2 without Assumption BA.
\begin{theorem} \label{asymptotic normality of UOJIVE2}
    Under Assumptions \ref{regularity}, \ref{finite XX}, and \ref{finite Xepsilon}, $\sqrt{N}(\hat{\beta}_{UOJIVE2} - \beta) \overset{d}{\to} N(0,\sigma_{\epsilon}^2H^{-1})$.
\end{theorem}

\subsection{Many-instrument consistency of UOJIVE2}
The many-instrument asymptotics framework is that both $K_1$ and $N$ go to infinity and the ratio $\frac{K_1}{N}$ converges to a constant $\alpha$ where $0 < \alpha < 1$. The motivation behind the many-instrument framework is to provide a better approximation of a situation where the number of instruments is large relative to sample size and first-stage overfitting is concerning. Many papers have looked into many-instrument setup \citep{bekker1994alternative, kunitomo2012optimal,chao2012asymptotic, hausman2012instrumental}. Theorem 1 from \citet{hausman2012instrumental} directly applies to UOJIVE2.
\begin{theorem} \label{theorem many instrument asymptotic for UOJIVE2}
    Under Assumption 1-4 specified in \citet{hausman2012instrumental}, $\hat{\beta}_{UOJIVE2} \overset{p}{\to} \beta$ as $N \to \infty$.
\end{theorem}

\section{Simulation} \label{simulation}
I run three types of simulations to contrast the performances of OLS, TSLS, JIVEs, and UIJIVEs (UOJIVEs). They are designed to test for many-instrument asymptotic under homoskedasticty, many-instrument asymptotic under consistency under heteroskedasticity, and robustness to outliers. Each simulation has two or four setups and each setup consists of 1000 rounds. All simulations have $L_1 = 1$. There is only one endogenous variable $X^*$. $\beta^*$ is set to be 0.3.  The intercepts are both stages are set to be 0 throughout all simulations, though a column of ones is still included in the regressions, either partialled out as part of $W$ for UIJIVEs or included as a part of $Z$ for the rest of the estimators. 

\subsection{Simulation for many-instrument asymptotic under homoskedasticity}
The simulation setup for many-instrument asymptotics is summarized in Table \ref{Simulation Setup Table}. The parameters for all instruments $Z^*$ (and for all controls $W$) are set as a constant, so I only report one value for each parameter instead of a vector. $R_0 = \frac{N\pi E[Z'Z] \pi}{\sigma_\eta^2}$ is the concentration parameter, I set it to be around 150 following \citet{hansen2014instrumental} to maintain a reasonable instrumental variable strength. $Z$ follows standard multivariate normal distributions.\footnote{I follow \citet{bekker2015jackknife} which also assumes nonrandom $Z$ for analyzing the theoretical property of the IV estimator, but let $Z$ follow a random distribution in their simulation study.}
The error terms $\epsilon$ and $\eta$ are bivariate normal with mean $(0,0)'$ and covariance matrix 
$\begin{pmatrix}
    0.8 & -0.6 \\
    -0.6 & 1 
\end{pmatrix}$.
	 
\begin{table}
    \centering
    \begin{tabular}{c c c c c c c c c}
            \hline
             Setup & $N$ & $K$ & $L$ & $\beta^*$ & $\gamma^*$ & $\pi^*$ & $\delta^*$ & $R_0$ \\
             \hline
             1 & 500 & 50 & 10 & 0.3 & 1 & 0.08 & 0.05 & 140.5 \\
             2 & 2000 & 200 & 40 & 0.3 & 1 & 0.02 & 0.02 & 160 \\
             \hline
        \end{tabular}
    \caption{The simulation setups for testing estimators' performances under many-instrument scenarios under homoskedasticity}
    \label{Simulation Setup Table}
\end{table}

I report the simulation results in Table~\ref{table simulation many-instrument}, it is clear that OLS, TSLS, and JIVEs suffer from a larger bias than approximately unbiased estimators. In addition, as I increase the degree of overidentification and the number of control variables, the bias of TSLS and JIVEs deteriorates. The sharp worsening of JIVEs' bias is likely exacerbated by its instability as pointed out by \citet{davidson2007moments}. The simulation results align with Corollary \ref{corollary for approximate bias} which implies that TSLS' approximate bias is proportional to the degree of overidentification while JIVEs' approximate biases are proportional to the number of control variables. All approximately unbiased estimators perform well under both setups. For example, UOJIVE1 and UOJIVE2 virtually do not show any bias and have a relatively small variance compared to JIVEs. 

\begin{table}[ht]
    \centering
    \begin{tabular}{lcccccc}
        \toprule
        & \multicolumn{3}{c}{Setup 1} & \multicolumn{3}{c}{Setup 2} \\
        \cmidrule(lr){2-4} \cmidrule(lr){5-7}
        Estimator & Bias & Variance & MSE & Bias & Variance & MSE \\
        \midrule
        OLS        & 0.475 & 0.001 & 0.226 & 0.564 & 0.000 & 0.318 \\
        TSLS       & 0.143 & 0.004 & 0.024 & 0.337 & 0.002 & 0.116 \\
        Nagar      & 0.013 & 0.009 & 0.009 & 0.055 & 0.016 & 0.019 \\
        AUK        & 0.007 & 0.011 & 0.011 & 0.023 & 0.027 & 0.027 \\
        JIVE1      & 0.069 & 0.016 & 0.021 & 0.421 & 1.009 & 1.185 \\
        JIVE2      & 0.069 & 0.016 & 0.021 & 0.420 & 1.093 & 1.268 \\
        TSJI1      & 0.003 & 0.010 & 0.010 & 0.002 & 0.022 & 0.022 \\
        TSJI2      & 0.003 & 0.010 & 0.010 & 0.002 & 0.022 & 0.022 \\
        UIJIVE1    & 0.003 & 0.010 & 0.010 & 0.002 & 0.023 & 0.023 \\
        UIJIVE2    & 0.003 & 0.010 & 0.010 & 0.002 & 0.023 & 0.023 \\
        UOJIVE1    & 0.003 & 0.010 & 0.010 & 0.002 & 0.023 & 0.023 \\
        UOJIVE2    & 0.003 & 0.010 & 0.010 & 0.002 & 0.022 & 0.022 \\
        \bottomrule
    \end{tabular}
    \caption{Comparison of estimator performances with many instruments}
    \label{table simulation many-instrument}
\end{table}

\subsection{Many-instrument asymptotic under heteroskedasticity}
Following the setup in \citet{ackerberg2009improved}, I test the performance of estimators with heteroskedastic errors while setting $Z$ to be a group-fixed effect matrix without any control variables $W$. Since there is no $W$, UIJIVEs and UOJIVEs are equivalent and I only report UOJIVEs' results.

The sample size is set to be 500, $N = 500$. The first 115 observations belong to Group 1, the next 115 observations belong to Group 2. Groups 1 and 2 are two big groups. The rest of the 270 observations consist of 18 small groups. Each small group has 15 observations. The first group is excluded from $Z$ and $\pi^* = 0.3$, meaning that Group 1 on average has its $X$ value 0.3 below other groups conditional on everything else being equal. There are two covariance matrices:
$\begin{pmatrix}
    0.25 & -0.1 \\
    -0.1 & 0.25
\end{pmatrix}$ and
$\begin{pmatrix}
    0.25 & 0.2 \\
    0.2 & 0.25
\end{pmatrix}$, denoted as $-$ and $+$ covariance matrices, respectively. In Setup 1, I let the big groups have the $+$ covariance matrix and the small groups have the $-$ covariance matrix; in Setup 2, I reverse the covariance matrices for the two types of groups.

I summarize the simulation results in Table~\ref{table simulation heteroskedasticity}. The result aligns with \citet{ackerberg2009improved}. The Nagar estimator is neither approximately unbiased nor consistent under heteroskedasticity with many instruments. The same can be said about approximately unbiased k-class estimator AUK which is proposed in Appendix~\ref{sec other classes of estimators}. In contrast, the performance of UOJIVE2 remains stellar, aligning with Theorem \ref{theorem many instrument asymptotic for UOJIVE2} that consistency of UOJIVE2 is established without assuming homoskedasticity. 

\begin{table}[ht]
    \centering
    \begin{tabular}{lcccccc}
        \toprule
        & \multicolumn{3}{c}{Setup 1: ${\rm small}^{+}{\rm big}^{-}$} & \multicolumn{3}{c}{Setup 2: ${\rm small}^{-}{\rm big}^{+}$} \\
        \cmidrule(lr){2-4} \cmidrule(lr){5-7}
        Estimator & Bias & Variance & MSE & Bias & Variance & MSE \\
        \midrule
        OLS        & 0.232 & 0.002 & 0.056 & 0.141 & 0.002 & 0.022 \\
        TSLS       & 0.286 & 0.028 & 0.109 & 0.135 & 0.025 & 0.043 \\
        Nagar      & 0.258 & 6.927 & 6.987 & 0.364 & 0.274 & 0.406 \\
        AUK        & 0.336 & 0.246 & 0.358 & 0.409 & 1.087 & 1.253 \\
        JIVE1      & 0.311 & 119.913 & 119.890 & 0.047 & 0.116 & 0.118 \\
        JIVE2      & 0.036 & 0.392 & 0.393 & 0.047 & 0.113 & 0.115 \\
        TSJI1      & 0.054 & 0.127 & 0.130 & 0.072 & 0.073 & 0.078 \\
        TSJI2      & 0.075 & 0.234 & 0.239 & 0.072 & 0.069 & 0.074 \\
        UOJIVE1    & 0.011 & 0.088 & 0.088 & 0.023 & 0.064 & 0.065 \\
        UOJIVE2    & 0.019 & 0.095 & 0.096 & 0.024 & 0.061 & 0.062 \\
        \bottomrule
    \end{tabular}
    \caption{Comparison of estimator performances under heteroskedasticity}
    \label{table simulation heteroskedasticity}
\end{table}

\subsection{Simulation with outlier}
In my proofs for the consistency and asymptotic normality results of UOJIVE1, Assumption BA is repeatedly invoked to bound below the denominator. With high leverage points, the asymptotic proofs for UOJIVE1 are no longer valid. In contrast, UOJIVE2's asymptotic results do not require Assumption BA. In the following simulation setups, I deliberately introduce high leverage points in the DGP, when these high leverage points coincide with large variances of $\epsilon$, the performance of UOJIVE1 is much worse than that of UOJIVE2. 

There are 5 instruments $Z^*$, one endogenous variable $X^*$, and no controls $W$. Again, UIJIVEs and UOJIVEs are equivalent due to the absence of $W$, so I only report UOJIVEs. I set the sample size $N$ to be $\{101,401,901,1601\}$. All these numbers are 1 plus a square number so it is easy to set up the following simulations. The first observation is the high leverage point with its first entry being $(N-1)^{1/3}$. For the rest of the $N-1$ observations, every $\sqrt{N-1}$ observations have their first five rows being the identity matrix, the rest of the rows are all zeroes. This setup is equivalent to the group fixed effect setup for the heteroskedasticity simulation where there are five small groups (indexed 2-6) and one large group (Group 1), except the first row which is supposed to belong to Group 2 is contaminated, and has its value multiplied by $N^{1/3}$. I include Table~\ref{table matrix with outlier} in the Appendix~\ref{sec matrix with outlier} to illustrate this simulation setup. The error terms $\epsilon$ and $\eta$ are bivariate normal with mean $(0,0)'$ and covariance matrix 
$\begin{pmatrix}
    0.8 & -0.6 \\
    -0.6 & 1 
\end{pmatrix}$, except that $\epsilon_1$, the error terms for the high leverage point, is multiplied by $N^{1/3}$. The coincidence between high leverage and large variance of $\epsilon$ generates an \textbf{outlier}. Intuitively, the first observation has a high influence on its fitted value and has a large probability of deviating a lot from the regression line. $\pi^*$ is set to be one.

The results for the simulations with outliers are summarized in Table \ref{table simulation outliers}. Throughout the four simulations, UOJIVE2 outperforms UOJIVE1 and TSJI2 outperforms TSJI1 substantially. The simulation result points to the importance of Assumption BA to the consistency of UOJIVE1. 
\begin{table}[ht]
\centering
\begin{tabular}{cc}
\toprule
\begin{tabular}{>{\centering\arraybackslash}p{0.45\linewidth}}
\textbf{N = 101} \\
\begin{tabular}{lccc}
\toprule
Estimator & Bias & Variance & MSE \\
\midrule
TSJI1   & 0.018 & 0.388 & 0.388 \\
TSJI2   & 0.014 & 0.130 & 0.130 \\
UOJIVE1 & 0.013 & 0.193 & 0.193 \\
UOJIVE2 & 0.001 & 0.067 & 0.067 \\
\end{tabular}
\end{tabular}
&
\begin{tabular}{>{\centering\arraybackslash}p{0.45\linewidth}}
\textbf{N = 401} \\
\begin{tabular}{lccc}
\toprule
Estimator & Bias & Variance & MSE \\
\midrule
TSJI1   & 0.054 & 0.397 & 0.400 \\
TSJI2   & 0.012 & 0.110 & 0.110 \\
UOJIVE1 & 0.036 & 0.169 & 0.170 \\
UOJIVE2 & 0.010 & 0.036 & 0.036 \\
\end{tabular}
\end{tabular} \\
\midrule
\begin{tabular}{>{\centering\arraybackslash}p{0.45\linewidth}}
\textbf{N = 901} \\
\begin{tabular}{lccc}
\toprule
Estimator & Bias & Variance & MSE \\
\midrule
TSJI1   & 0.029 & 0.359 & 0.360 \\
TSJI2   & 0.010 & 0.093 & 0.092 \\
UOJIVE1 & 0.019 & 0.144 & 0.144 \\
UOJIVE2 & 0.004 & 0.024 & 0.024 \\
\end{tabular}
\end{tabular}
&
\begin{tabular}{>{\centering\arraybackslash}p{0.45\linewidth}}
\textbf{N = 1601} \\
\begin{tabular}{lccc}
\toprule
Estimator & Bias & Variance & MSE \\
\midrule
TSJI1   & 0.034 & 0.395 & 0.396 \\
TSJI2   & 0.004 & 0.097 & 0.097 \\
UOJIVE1 & 0.019 & 0.152 & 0.152 \\
UOJIVE2 & 0.003 & 0.021 & 0.020 \\
\end{tabular}
\end{tabular} \\
\bottomrule
\end{tabular}
\caption{Comparison of estimator performances with outliers}
\label{table simulation outliers}
\end{table}

Moreover, the consistency results of UOJIVE2 seem to be preserved under this simulation setup even though $\Sigma_{ZZ}^{-1}$ does not exist. As shown by Figure \ref{fig outliers}, the empirical distribution of the 1000 UOJIVE2 estimates concentrates around $\beta^* = 0.3$ more and more as the sample size increases, whereas that of UOJIVE1 does not change much across the four simulations. It can be an interesting future research direction to show that UOJIVE2 is consistent under such a weak-instrument setup whereas UOJIVE1 is not.

\begin{figure}[!htb]
  \centering
  \subfloat[$N = 101$]{\includegraphics[width=0.45\textwidth]{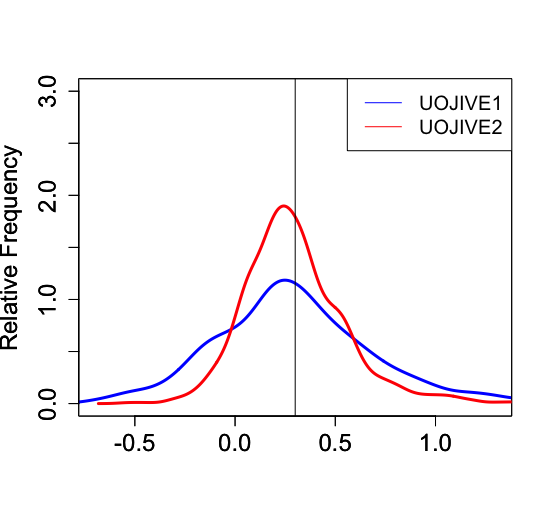}\label{fig:outlier100}}
  \hfill 
  \subfloat[$N = 401$]{\includegraphics[width=0.45\textwidth]{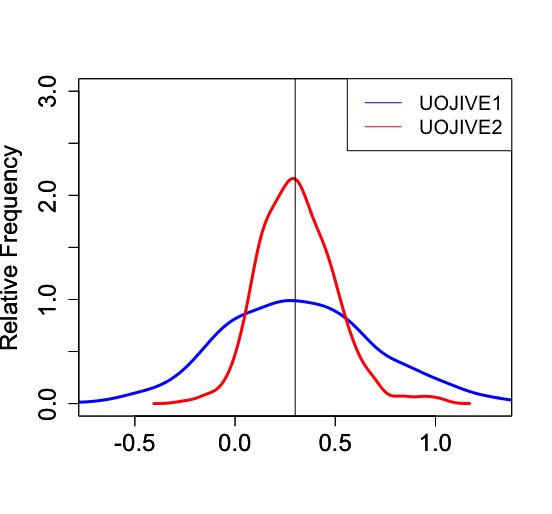}\label{fig:outlier400}}
  
  \subfloat[$N = 901$]{\includegraphics[width=0.45\textwidth]{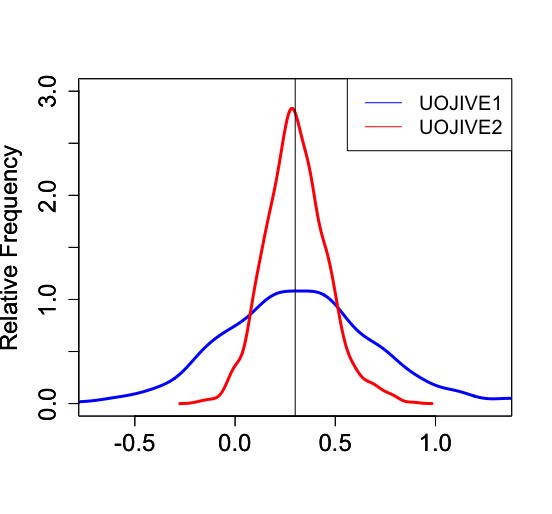}\label{fig:outlier900}}
  \hfill 
  \subfloat[$N = 1601$]{\includegraphics[width=0.45\textwidth]{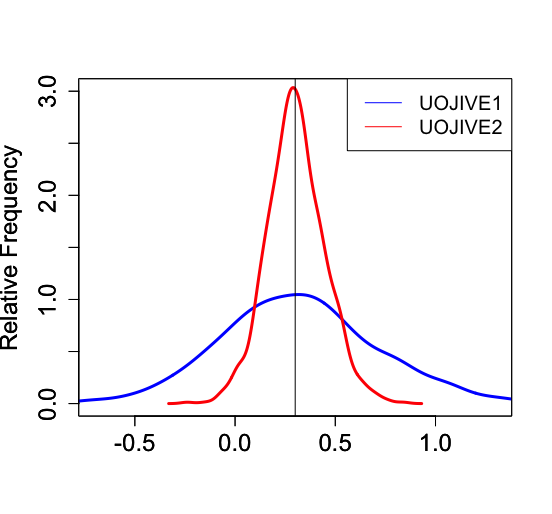}\label{fig:outlier1600}}
  \caption{Simulation with outliers}
  \label{fig outliers}
\end{figure}

\section{Empirical Studies} \label{Empirical}
There are multitudes of social science studies that use a large number of instruments. Some examples include the judge leniency IV design where researchers use the identity of the judge as an instrument. In other words, the number of instruments is equal to the number of judges in the sample. The method has been applied to other settings (See Table 1 in \citet{frandsen2023judging} for the immense popularity of judge leniency design). In this section, I apply approximately unbiased estimators to two classical empirical studies. I compute the standard error by assuming homoskedasticity and treating these approximately unbiased estimators as just-identified IV estimators.

\subsection{Quarter of birth} 
The quarter of birth example has been repeatedly cited by many-instrument literature. Here I apply TSJIs and UOJIVEs the famous example in \citet{angrist1991does}. 

Many states in the US have a compulsory school attendance policy. Students are mandated to stay in school until their 16th, 17th, or 18th birthday depending on which state they are from. As such, students' quarter of birth may induce different quitting-school behaviors. This natural experiment makes a quarter of birth a valid IV to estimate the marginal earnings brought by an additional school year for those who are affected by the compulsory attendance policy. 

\citet{angrist1991does} interacts quarter of birth with other dummy variables to generate a large number of IVs,
\begin{enumerate}
    \item Quarter of birth $\times$ Year of birth 
    \item Quarter of birth $\times$ Year of birth $\times$ State of birth
\end{enumerate}
where case 1 contains 30 instruments, and case 2 contains 180 instruments. The results are reported in Table \ref{return of school}.

\begin{table}[ht]
\centering
    \begin{tabular}{lcccc}
        \toprule
        \textbf{Estimator} & \multicolumn{2}{c}{\textbf{Case 1}} & \multicolumn{2}{c}{\textbf{Case 2}} \\
        \cmidrule(lr){2-3} \cmidrule(lr){4-5}
        & \textbf{Estimate (\%)} & \textbf{ SE (\%)} & \textbf{Estimate (\%)} & \textbf{ SE (\%)} \\
        \midrule
        TSLS          & 8.91 & 1.61 & 9.28 & 0.93 \\
        Nagar         & 9.35 & 1.80 & 10.88 & 1.20 \\
        AUK           & 9.35 & 1.80 & 10.88 & 1.20 \\
        JIVE1         & 9.59 & 1.90 & 12.11 & 1.37 \\
        JIVE2         & 9.59 & 1.90 & 12.11 & 1.37 \\
        TSJI1         & 9.34 & 1.79 & 10.92 & 1.20 \\
        TSJI2         & 9.34 & 1.79 & 10.93 & 1.20 \\
        UIJIVE1       & 9.34 & 1.80 & 10.20 & 1.08 \\
        UIJIVE2       & 9.34 & 1.80 & 10.95 & 1.20 \\
        UOJIVE1       & 9.34 & 1.80 & 10.20 & 1.08 \\
        UOJIVE2       & 9.34 & 1.80 & 10.95 & 1.20 \\
        \bottomrule
    \end{tabular}
    \caption{text}
    \label{return of school}
\end{table}

\subsection{Veteran's smoking behavior}
\citet{bedard2006long} use year of birth and its interaction with gender as instruments to estimate how much enlisting for WWII and the Korean War increases the veterans' probability of smoking during the later part of their lives. The result can be interpreted as LATE.
\begin{enumerate}
    \item Birth year $\times$ gender
    \item Birth year
\end{enumerate}
where case 1 uses all data and case 2 uses only data for male veterans. The results are summarized in Table \ref{table smoking}.

\begin{table}[ht]
\centering
\begin{tabular}{lcccccccc}
\toprule
\textbf{Estimator} & \multicolumn{4}{c}{\textbf{Case 1}} & \multicolumn{4}{c}{\textbf{Case 2}} \\
\cmidrule(lr){2-5} \cmidrule(lr){6-9}
& \multicolumn{2}{c}{\textbf{CPS 60}} & \multicolumn{2}{c}{\textbf{CPS 90}} & \multicolumn{2}{c}{\textbf{CPS 60}} & \multicolumn{2}{c}{\textbf{CPS 90}} \\
\cmidrule(lr){2-3} \cmidrule(lr){4-5} \cmidrule(lr){6-7} \cmidrule(lr){8-9}
& \textbf{Estimate} & \textbf{SE} & \textbf{Estimate} & \textbf{SE} & \textbf{Estimate} & \textbf{SE} & \textbf{Estimate} & \textbf{SE} \\
\midrule
TSLS         & 27.63 & 3.50 & 34.58 & 2.38 & 23.69 & 13.90 & 30.14 & 3.15 \\
Nagar        & 27.82 & 3.51 & 34.68 & 2.39 & 32.42 & 17.61 & 30.43 & 3.18 \\
AUK          & 27.82 & 3.51 & 34.68 & 2.39 & 32.45 & 17.62 & 30.43 & 3.18 \\
JIVE1        & 28.50 & 3.65 & 35.03 & 2.43 & -136.12 & 224.30 & 31.13 & 3.33 \\
JIVE2        & 28.50 & 3.65 & 35.03 & 2.43 & -136.12 & 224.35 & 31.13 & 3.33 \\
TSJI1        & 27.84 & 3.53 & 34.68 & 2.39 & 32.52 & 21.99 & 30.43 & 3.21 \\
TSJI2        & 27.84 & 3.53 & 34.68 & 2.39 & 32.52 & 21.99 & 30.43 & 3.21 \\
UIJIVE1      & 27.85 & 0.12 & 34.70 & 2.39 & 32.37 & 17.64 & 30.43 & 3.18 \\
UIJIVE2      & 27.85 & 0.12 & 34.70 & 2.39 & 32.37 & 17.65 & 30.43 & 3.18 \\
UOJIVE1      & 27.85 & 0.12 & 34.70 & 2.39 & 32.37 & 17.64 & 30.43 & 3.18 \\
UOJIVE2      & 27.85 & 0.12 & 34.70 & 2.39 & 32.37 & 17.65 & 30.43 & 3.18 \\
\bottomrule
\end{tabular}
\caption{Estimation Results from Empirical Study for Cases 1 and 2, CPS 60 and 90}
\label{table smoking}
\end{table}

The results of all estimators are close except for Case 2 with CPS90. In this setup of Case 2 CPS90, JIVE's result, which is negative and counterintuitively large in magnitude (larger than 1), is driven by its instability. TSLS' estimate is about 10\% below the estimates of other approximately unbiased estimators. Compared to other setups (Case 1 CPS60, Case 1 CPS90, Case 2 CPS90), the TSLS estimate for this case is the smallest. In contrast, for other approximately unbiased estimators, their estimates throughout all four cases are closer to each other. \citet{bedard2006long} claims that the four TSLS estimates ``are almost identical''. Table \ref{table smoking} gives stronger evidence for this claim with even closer estimates from other approximately unbiased estimators. 

\section{Conclusion} \label{conclusion}
This paper formalizes the definition of approximate bias and extends its applicability to other estimators that the approximate bias literature has not considered. The extension motivates new approximately unbiased estimators such as UOJIVE2. I show that UOJIVE2 is consistent under a fixed number of instruments and control variables as the sample size goes to infinity. The consistency and asymptotic normality results of UOJIVE2 do not require the absence of high leverage points, a condition that is necessary for the proofs that I construct to establish consistency and asymptotic normality of UOJIVE1. The simulation study aligns with these theoretical results. When a high leverage point coincides with a high variance of the error term, an outlier is generated and the performance of UOJIVE1 is much poorer than that of UOJIVE2.

%

\newpage
\begin{appendices}

\section{Approximate bias for classes of estimators that have $CZ = Z$ property}	\label{def for approximate bias derivation}
Recall that 
\begin{equation*}
	R_N = \underbrace{J\epsilon}_{R1} - \underbrace{\frac{Q_0}{N}\pi'Z'\eta J \epsilon}_{R2} +\underbrace{\frac{Q_0}{N}\eta' C'\epsilon}_{R3} - \underbrace{\frac{Q_0}{N}\eta'P_{Z\pi}\epsilon}_{R4} + o_P(\frac{1}{N})
\end{equation*}
In this section, I will provide proof for corollary \ref{corollary for approximate bias} and in the process, the derivation of $R1$, $R2$, $R3$ and $R4$. 
Consider an IV estimator that takes form of $(X'C'X)^{-1}(X'C'y) = \beta + (X'C'X)^{-1}(X'C'\epsilon)$ where $CZ= Z$ and hence, $CX = Z\pi+C\eta$.
\begin{align*}
	(X'CX)^{-1}(X'C\epsilon) = & (X'CX)^{-1}(X'C\epsilon)\\
	= & ((\pi'Z'+\eta'C')X)^{-1}(X'C'\epsilon) \\
	= & (\pi'Z'X+\eta'C'X)^{-1}(X'C'\epsilon) \\
	= & (I + Q\eta'C'X)^{-1}Q(X'C'\epsilon) \quad \text{where} \quad Q = (\pi'Z'X)^{-1}\\
	= & (I - \underbrace{Q\eta'C'X}_{\sim O_P(\frac{1}{\sqrt{N}})})\underbrace{Q(X'C'\epsilon)}_{\sim O_P(\frac{1}{\sqrt{N}})} + o_P(\frac{1}{N})
\end{align*}
The last step is a geometric expansion of $(I + Q\eta'C'X)^{-1}$ where $Q\eta'C'X = O_P(\frac{1}{\sqrt{N}})$ since $Q = (\pi'Z'X)^{-1} = O_P(\frac{1}{N})$ and $\eta'C'X =O_P(\sqrt{N})$. The first term's stochastic order is obvious, I evaluate the stochastic order of $\lambda_1$-class estimator's $\eta'C'X$ as an example to show that $\eta'C'X =O_P(\sqrt{N})$. The proof for $\omega_1$-class is similar but easier given that $\omega \sim O(\frac{1}{N})$ and Assumption BA. The proof for the k-class estimator is trivial.
\begin{align*}
	\eta'C'X = & \eta' (P_Z - \lambda D)(I - \lambda D)^{-1}X \\
	= & \eta' Z(Z'Z)^{-1}Z' (I - \lambda D)^{-1} X - \lambda  \sum_{i=1}^{N}\eta_i'Z_i(Z'Z)^{-1}Z_i' \frac{X_i}{1-\lambda D_i} \\
	= & \eta' Z(Z'Z)^{-1}\sum_{i=1}^{N} \frac{Z_i'X_i}{1-\lambda D_i}- \lambda  \sum_{i=1}^{N}\eta_i'D_i \frac{X_i}{1-\lambda D_i}
\end{align*}

I make the following Assumption BA for the leverage of projection matrix $P_Z = Z(Z'Z)^{-1}Z'$, the assumption implies that for large enough $N$, $D_i \leq m$ for some fixed $m > 0$, for all $i = 1,2,3 \dots, N$.

\begin{lemma}\label{etaCX A}
    Assume that $0\leq\lambda\leq 1$, $\eta' Z(Z'Z)^{-1}\sum_{i=1}^{N} \frac{Z_i'X_i}{1-\lambda D_i} = O_P({\sqrt{N}})$.
\end{lemma}
\begin{proof}
	\begin{equation*}
		\eta' Z(Z'Z)^{-1}\sum_{i=1}^{N} \frac{Z_i'X_i}{1-\lambda D_i}= O_P(\sqrt{N})O_P(\frac{1}{N})O_P(N) = O_P(\sqrt{N})
	\end{equation*}
because CLT applies to $\frac{1}{\sqrt{N}}\eta' Z$ and law of large of number applies to the summation term when divided by $N$.
\end{proof}

\begin{lemma}\label{lemma moment to max order}
    $E[\lVert \phi_i \rVert^{r+\delta_1}] \text{ for some $\delta_1 > 0$} \implies \frac{1}{N^{1/r}} \max_i{\lVert \phi_i \rVert} = o_P(1)$.
\end{lemma}
\begin{proof}
For a fixed $\delta > 0$,
\begin{align*}
    P( \frac{1}{N^{1/r}} \max_i{\lVert \phi_i \rVert} < \delta) = & P(  \max_i \lVert \phi_i \rVert < \delta N^{1/r} ) \\
	= & P(\lVert \phi_i \rVert < \delta N^{1/r} \quad \text{for} \quad i=1,2,\dots,N) \\
	= & P(\lVert \phi_i \rVert < \delta N^{1/r})^N \\
	= & P(\lVert \phi_i \rVert^{r+\delta_1} < (\delta N^{1/r})^{r+\delta_1} )^N \quad \text{for $\delta_1 > 0$} \\
	= & \left(1 - P(\lVert \phi_i \rVert^{r+\delta_1} \geq (\delta N^{1/r})^{r+\delta_1} )\right)^N \\
	\geq & \left(1 - \frac{E[\lVert \phi_i \rVert^{r+\delta_1}] }{(\delta N^{1/r})^{r+\delta_1}}  \right)^N \\
	= & \left(1 - \frac{1}{N}\frac{E[\lVert \phi_i \rVert^{r+\delta_1}] }{\delta^{r+\delta_1}N^{\delta_1/r}}  \right)^N \\
	\geq & 1 - \frac{E[\lVert \phi_i \rVert^{r+\delta_1}] }{\delta^{r+\delta_1}N^{\delta_1/r}} \to 1
\end{align*}

The last inequality holds when $N\geq 1$ and $\frac{E[\lVert \phi_i \rVert^{r+\delta_1}] }{\delta^{r+\delta_1}N^{\delta_1/r}} < N$, both of which are true for large $N$.
\end{proof}

With Lemma~\ref{lemma moment to max order}, it is easy to show the following: assuming that $E[\lVert \eta_i'X_i\rVert^{2+\delta_1}]$ is finite, for some $\delta_1 > 0$, $\lVert \sum_{i=1}^{N}\eta_i'D_i \frac{X_i}{1-\lambda D_i} \rVert \leq \frac{K}{1-m} \max_i \lVert \eta_i'X_i \rVert = o_P(\sqrt{N})$. Combining this result with Lemma~\ref{etaCX A}, I show that $\eta'C'X = O_P(\sqrt{N})$. The proof for $X'C'\epsilon = O_P(\sqrt{N})$ is similar.

Next, I sub $X = Z\pi + \eta$ into the $(I - {Q\eta'C'X})Q(X'C'\epsilon)+ o_P(\frac{1}{N})$.
\begin{align*}
	(X'CX)^{-1}(X'C\epsilon) = & (I - {Q\eta'C'(Z\pi + \eta)}Q((Z\pi + \eta)'C'\epsilon)+ o_P(\frac{1}{N}) \\
	=& (I - Q\eta'C'Z\pi - Q'\eta'C'\eta)Q(\pi'Z'\epsilon + \eta'C'\epsilon) + o_P(\frac{1}{N}) \\
	= &Q\pi'Z'\epsilon + Q\eta'C'\epsilon - Q\eta'C'Z\pi Q\pi'Z'\epsilon + o_P(\frac{1}{N}) 
\end{align*}
The last equality holds because after cross-multiplying, we have six terms to evaluate:
\begin{table}[H]
	\centering
	\renewcommand{\arraystretch}{1}
	\begin{tabular}{c c c}
		term & stochastic order & keep or not \\
		$Q\pi'Z'\epsilon$ &  $O_P(\frac{1}{\sqrt{N}})$& Yes\\ 
		$Q\eta'C'\epsilon$ & $O_P(\frac{1}{N})$ & Yes\\ 
		$-Q\eta'C'Z\pi Q\pi'Z'\epsilon$ & $O_P(\frac{1}{N})$& Yes\\
		$-Q\eta'C'Z\pi Q \eta'C'\epsilon$ & $O_P(\frac{1}{N\sqrt{N}})$& No\\ 
		$-Q\eta'C'\eta Q\pi'Z'\epsilon$ & $O_P(\frac{1}{N\sqrt{N}})$& No\\
		$-Q\eta'C'\eta Q\eta'C'\epsilon$ & $O_P(\frac{1}{N^2})$ & No
	\end{tabular}
\end{table}

After dropping the last three terms, we obtain the following expression for the difference between the estimator and $\beta$:
\begin{equation} \label{intermediate equation for evaluating approximate bias}
	(X'CX)^{-1}(X'C\epsilon) = Q\pi'Z'\epsilon+ Q\eta'C'\epsilon-Q\eta'C'Z\pi Q\pi'Z'\epsilon + o_P(\frac{1}{N}) 
\end{equation}

We evaluate the three terms in Eq.(\ref{intermediate equation for evaluating approximate bias}) separately.

\subsection{$Q\pi'Z'\epsilon$}\label{A1}
\begin{align*}
	Q\pi'Z'\epsilon =& (\pi'Z'X)^{-1}\pi'Z'\epsilon\\
	=& (\pi'Z'Z\pi+\pi'Z'\eta)^{-1}\pi'Z'\epsilon\\
	=& \underbrace{(\pi'Z'Z\pi)^{-1}\pi'Z'\epsilon}_{E[(\pi'Z'Z\pi)^{-1}\pi'Z'\epsilon]=0} - (\pi'Z'Z\pi)^{-1}\pi'Z'\eta(\pi'Z'Z\pi+\pi'Z'\eta)^{-1}\pi'Z'\epsilon
\end{align*}
The part of the expression with a zero expectation is exactly $R1$. 

\begin{align*}
	&(\pi'Z'Z\pi)^{-1}\pi'Z'\eta(\pi'Z'Z\pi+\pi'Z'\eta)^{-1}\pi'Z'\epsilon \\
	= &(\pi'Z'Z\pi)^{-1}\pi'Z'\eta(\pi'Z'Z\pi)^{-1}\pi'Z'\epsilon- \underbrace{(\pi'Z'Z\pi)^{-1}\pi'Z'\eta(\pi'Z'Z\pi)^{-1}\pi'Z'\eta(\pi'Z'Z\pi+\pi'Z'\eta)^{-1}\pi'Z'\epsilon}_{\sim O_P(\frac{1}{N\sqrt{N}})}\\
	= & \underbrace{\frac{Q_0}{N}\pi'Z'\eta(\pi'Z'Z\pi)^{-1}\pi'Z'\epsilon}_{R2}+ o_P(\frac{1}{N})\\
\end{align*}
The last equality holds because $N(\pi'Z'Z\pi)^{-1} {\to} Q_0$, therefore, $(\pi'Z'Z\pi)^{-1}- \frac{Q_0}{N} = o(\frac{1}{N})$. And we also have that $\pi'Z'\eta(\pi'Z'Z\pi)^{-1}\pi'Z'\epsilon = O_P(1)$. So,
\begin{align*}
	& \frac{Q_0}{N}\pi'Z'\eta(\pi'Z'Z\pi)^{-1}\pi'Z'\epsilon - (\pi'Z'Z\pi)^{-1}\pi'Z'\eta(\pi'Z'Z\pi)^{-1}\pi'Z'\epsilon \\
 =& (\frac{Q_0}{N} - (\pi'Z'Z\pi)^{-1})\pi'Z'\eta(\pi'Z'Z\pi)^{-1}\pi'Z'\epsilon = o(\frac{1}{N})O_P(1) = o_P(\frac{1}{N}).
\end{align*} 
Then, I evaluate the expectation of R2:
\begin{align*}
	E[\frac{Q_0}{N}\pi'Z'\eta(\pi'Z'Z\pi)^{-1}\pi'Z'\epsilon] =& \frac{Q_0}{N}E[\pi'Z'E[\eta(\pi'Z'Z\pi)^{-1}\pi'Z'\epsilon|Z]] \\
	=&\frac{Q_0}{N} E[\pi'Z'((\pi'Z'Z\pi)^{-1}\pi'Z')']\sigma_{\eta\epsilon} \\
	=& \frac{Q_0}{N} I_L \sigma_{\eta\epsilon} = \frac{Q_0}{N}\sigma_{\eta\epsilon}
\end{align*}

\subsection{$Q\eta'C'\epsilon$} \label{A2}
\begin{align*}
	&Q\eta'C'\epsilon = \underbrace{\frac{Q_0}{N}\eta'C'\epsilon}_{R3} + o_P(\frac{1}{N}) \\
	&E[\frac{Q_0}{N}\eta'C'\epsilon] = \frac{Q_0}{N}tr(C')\sigma_{\eta\epsilon}
\end{align*}

\subsection{$Q\eta'C'Z\pi Q\pi'Z'\epsilon$}\label{A3}
\begin{align*}
	Q\eta'C'Z\pi Q\pi'Z'\epsilon =& Q\eta'C'Z\pi (\pi'Z'X)^{-1}\pi'Z'\epsilon\\
	 =& Q\eta'C'Z\pi (\pi'Z'Z\pi +\pi'Z'\eta)^{-1}\pi'Z'\epsilon\\
	 =& Q\eta'C'Z\pi (\pi'Z'Z\pi)^{-1}\pi'Z'\epsilon   
	-
	 \underbrace{Q\eta'C'Z\pi (\pi'Z'Z\pi )^{-1} \pi'Z'\eta (\pi'Z'Z\pi +\pi'Z'\eta)^{-1}\pi'Z'\epsilon}_{\sim O(\frac{1}{N\sqrt{N}}) }\\
	 =&\underbrace{\frac{Q_0}{N}\eta'C'Z\pi (\pi'Z'Z\pi)^{-1}\pi'Z'\epsilon}_{\text{equivalent to $R4$ for approximate bias computation purpose}}+ o_P(\frac{1}{N})\\
\end{align*}
Though the last expression is not the same as $R4$, it does not affect the definition of approximate bias since we are only interested in the expectation of the last expression and that of $R4$. As long as the last expression and $R4$ share the same expectation, definition \ref{def for approximate bias} remains valid. Recall that $R4 = \frac{Q_0}{N}\eta'P_{Z\pi}\epsilon$ and $E[\frac{Q_0}{N}\eta'P_{Z\pi}\epsilon] = \frac{Q_0}{N}tr(P_{Z\pi}) \sigma_{\eta\epsilon}$. The following shows that the last expression has the same expectation.

\begin{align*}
	E[ \frac{Q_0}{N}\eta'C'Z\pi (\pi'Z'Z\pi)^{-1}\pi'Z'\epsilon] = & \frac{Q_0}{N}tr(C'Z\pi (\pi'Z'Z\pi)^{-1}\pi'Z') \sigma_{\eta\epsilon}\\
	=& \frac{Q_0}{N}tr(Z'C'Z\pi (\pi'Z'Z\pi)^{-1}\pi') \sigma_{\eta\epsilon}\\
	=& \frac{Q_0}{N}tr(Z'Z\pi (\pi'Z'Z\pi)^{-1}\pi') \sigma_{\eta\epsilon}\\
	=& \frac{Q_0}{N}tr(Z\pi (\pi'Z'Z\pi)^{-1}\pi'Z') \sigma_{\eta\epsilon}\\
	=& \frac{Q_0}{N}tr(P_{Z\pi}) \sigma_{\eta\epsilon}
\end{align*}

Combining the results from sections \ref{A1}, \ref{A2} and \ref{A3}, we obtain Corollary \ref{corollary for approximate bias}.

\section{Approximate bias for classes of estimators that do not have $CZ = Z$ property} \label{Approximate bias for other classes of estimators}

This section shows that definition \ref{def for approximate bias} applies to $\omega_2$-class and $\lambda_2$-class estimators. Once the validity of the definition is established, it is trivial to show that corollary \ref{corollary for approximate bias} is also true for these two classes of estimators.

\subsection{$\omega_2$-class estimators}
Recall that the closed-form expression of $\omega_2$-class estimator is 
\begin{equation*}
    \hat{\beta}_{\omega_2} = (X'(P_Z - D + \omega_2 I)'X)^{-1}(X'(P_Z - D + \omega_2 I)'y)
\end{equation*}
and the difference between $\hat{\beta}_{\omega_2}$ and $\beta$ is
\begin{align*}
    & (X'(P_Z - D + \omega_2 I)'X)^{-1}(X'(P_Z - D + \omega_2 I)'\epsilon)\\
    = & (QX'(P_Z - D + \omega_2 I)'X)^{-1}(QX'(P_Z - D + \omega_2 I)'\epsilon)\\
    = & (Q(\pi'Z' + \eta')(P_Z - D + \omega_2 I)'X)^{-1}(QX'(P_Z - D + \omega_2 I)'\epsilon) \\
    = & (I - \underbrace{Q\pi'Z'D'X}_{o_P(\frac{1}{\sqrt{N}})} + \underbrace{\omega_2 I}_{O(\frac{1}{N})} + \underbrace{Q\eta'P_Z'X}_{O_P(\frac{1}{\sqrt{N}})} - \underbrace{Q\eta'D'X}_{o_P(\frac{1}{\sqrt{N}})} + \underbrace{\omega_2Q\eta'X}_{O_P(\frac{1}{N})})^{-1}\underbrace{(QX'(P_Z - D + \omega_2 I)'\epsilon)}_{\sim O_P(\frac{1}{\sqrt{N}})}\\
    = & (I - Q\eta'P_Z'X)(QX'(P_Z - D + \omega_2 I)'\epsilon) + o_P(\frac{1}{N})
\end{align*}

After cross-multiplying, we obtain the following six terms

\begin{table}[H]
    \centering
    \renewcommand{\arraystretch}{1}
    \begin{tabular}{c c c}
        term & stochastic order & keep or not\\
        $QX'P_Z'\epsilon$ & $O_P(\frac{1}{\sqrt{N}})$ & Yes \\
        $-QX'D'\epsilon$ & $o_P(\frac{1}{\sqrt{N}})$ & Yes \\
        $\omega_2QX'\epsilon$ & $O_P(\frac{1}{N})$ & Yes \\
        $-Q\eta'P_Z'XQX'P_Z'\epsilon$ & $O_P(\frac{1}{N})$ & Yes \\ 
        $Q\eta'P_Z'XQX'D'\epsilon$ & $o_P(\frac{1}{N})$ & No\\
        $-\omega_2Q\eta'P_Z'XQX'\epsilon$ & $O_P(\frac{1}{N\sqrt{N}})$ & No
    \end{tabular}
\end{table}

\subsubsection{$QX'P_Z'\epsilon$}
\begin{align*}
    QX'P_Z'\epsilon =& (\pi'Z'X)^{-1}(\pi'Z' + \eta')P_Z'\epsilon\\
    =& (\pi'Z'X)^{-1}(\pi'Z'\epsilon + \eta'P_Z'\epsilon)\\
    =& \underbrace{(\pi'Z'Z\pi)^{-1}}_{O(\frac{1}{N})}(\underbrace{\pi'Z'\epsilon}_{O_P(\sqrt{N})} + \underbrace{\eta'P_Z'\epsilon}_{O_P(1)}) -
    \underbrace{(\pi'Z'Z\pi)^{-1}(\pi'Z'\eta)(\pi'Z'X)^{-1}}_{O_P(\frac{1}{N\sqrt{N}})}(\underbrace{\pi'Z'\epsilon}_{O_P(\sqrt{N})} + \underbrace{\eta'P_Z'\epsilon}_{O_P(1)})\\
    =& \underbrace{J\epsilon}_{R1} + \underbrace{\frac{Q_0}{N}\eta'P_Z'\epsilon}_{(a)} - \underbrace{\frac{Q_0}{N}(\pi'Z'\eta)(\pi'Z'Z\pi)^{-1}\pi'Z'\epsilon}_{R2} + o_P(\frac{1}{N})
\end{align*}

\subsubsection{$-QX'D'\epsilon$}
\begin{align*}
    -QX'D'\epsilon = & -(\pi'Z'Z\pi)^{-1}(\pi'Z' + \eta')D'\epsilon + \underbrace{(\pi'Z'Z\pi)^{-1}(\pi'Z'\eta)(\pi'Z'X)^{-1}}_{O_P(\frac{1}{N\sqrt{N}})}\underbrace{X'D'\epsilon}_{o_P({\sqrt{N}})}\\
    = & -\underbrace{(\pi'Z'Z\pi)^{-1}(\pi'Z'D'\epsilon)}_{E[(\pi'Z'Z\pi)^{-1}(\pi'Z'D'\epsilon)]=0} - \underbrace{\frac{Q_0}{N}(\eta'D'\epsilon)}_{(b)} + o_P(\frac{1}{N})
\end{align*}

\subsubsection{$\omega_2QX'\epsilon$}
\begin{align*}
    \underbrace{\omega_2QX'\epsilon}_{O_P(\frac{1}{N})} = & \omega_2(\pi'Z'Z\pi)^{-1}(\pi'Z' + \eta')\epsilon + o_P(\frac{1}{N}) \\
    =& \underbrace{\omega_2(\pi'Z'Z\pi)^{-1}(\pi'Z'\epsilon)}_{E[\omega_2(\pi'Z'Z\pi)^{-1}(\pi'Z'\epsilon)]=0} + \underbrace{\frac{Q_0}{N}\eta'(\omega_2 I)'\epsilon}_{(c)} + o_P(\frac{1}{N})  
\end{align*}
Note that $R3 = (a) - (b) + (c)$.

\subsubsection{$-Q\eta'P_Z'XQX'P_Z'\epsilon$}
\begin{align*}
    -\underbrace{Q\eta'P_Z'XQX'P_Z'\epsilon}_{O_P(\frac{1}{N})} =& -(\pi'Z'Z\pi)^{-1}\eta'Z\pi(\pi'Z'Z\pi)^{-1}\pi'Z'\epsilon + o_P(\frac{1}{N})\\
    =& -\underbrace{\frac{Q_0}{N}\eta'P_{Z\pi}\epsilon}_{R4} + o_P(\frac{1}{N})
\end{align*}

\subsection{$\lambda_2$-class estimators}
Recall that the closed-form expression of $\omega_2$-class estimator is 
\begin{equation*}
    \hat{\beta}_{\lambda_2} = (X'(P_Z - \lambda_2 D)'X)^{-1}(X'(P_Z - \lambda_2 D)'y)
\end{equation*}
and the difference between $\hat{\beta}_{\lambda_2}$ and $\beta$ is
\begin{align*}
    & (X'(P_Z - \lambda_2 D)'X)^{-1}(X'(P_Z - \lambda_2 D)'\epsilon)\\
    =& (QX'(P_Z - \lambda_2 D)'X)^{-1}(QX'(P_Z - \lambda_2 D)'\epsilon)\\
    =& (I + \underbrace{Q\eta'P_Z'X}_{O_P(\frac{1}{\sqrt{N}})} -\underbrace{\lambda_2QX'D'X}_{o_P(\frac{1}{\sqrt{N}})})^{-1}\underbrace{(QX'(P_Z - \lambda_2 D)'\epsilon)}_{O_P(\frac{1}{\sqrt{N}})}\\
    =& (I - \underbrace{Q\eta'P_Z'X}_{O_P(\frac{1}{\sqrt{N}})})(\underbrace{QX'P_Z'\epsilon}_{O_P(\frac{1}{\sqrt{N}})} - \underbrace{\lambda_2 QX'D'\epsilon}_{o_P(\frac{1}{\sqrt{N}})}) + o_P(\frac{1}{N})\\
    =& QX'P_Z'\epsilon - Q\eta'P_Z'XQX'P_Z'\epsilon - \lambda_2 QX'D'\epsilon + o_P(\frac{1}{N})
\end{align*}

\subsubsection{$QX'P_Z'\epsilon$}
\begin{align*}
    QX'P_Z'\epsilon = & Q\pi'Z'\epsilon + Q\eta'P_Z'\epsilon\\
    = & \underbrace{(\pi'Z'Z\pi)^{-1}\pi'Z'\epsilon}_{R1} +\underbrace{(\pi'Z'Z\pi)^{-1}\eta'Z\pi(\pi'Z'X)^{-1}\pi'Z'\epsilon}_{O_P(\frac{1}{N})} + \underbrace{Q\eta'P_Z'\epsilon}_{O_P(\frac{1}{N})}\\
    = & \underbrace{J\epsilon}_{R1} + \underbrace{\frac{Q_0}{N}\pi'Z'\eta(\pi'Z'Z\pi)^{-1}\pi'Z'\epsilon}_{R_2} + \underbrace{\frac{Q_0}{N}\eta'P_Z'\epsilon}_{(d)} + o_P(\frac{1}{N})
\end{align*}

\subsubsection{$-Q\eta'P_Z'XQX'P_Z'\epsilon$}
\begin{align*}
    -\underbrace{Q\eta'P_Z'XQX'P_Z'\epsilon}_{O_P(\frac{1}{N})} = & -\underbrace{(\pi'Z'Z\pi)^{-1}\eta'P_Z'(Z\pi + \eta))(\pi'Z'Z\pi)^{-1}(\pi'Z' + \eta')P_Z'\epsilon}_{O_P(\frac{1}{N})} + o_P(\frac{1}{N})\\
    = & -\underbrace{\frac{Q_0}{N}\eta'Z\pi(\pi'Z'Z\pi)^{-1}\pi'Z'\epsilon}_{R4} + o_P(\frac{1}{N})
\end{align*}

\subsubsection{$-\lambda_2 QX'D'\epsilon$}
\begin{align*}
    -\underbrace{\lambda_2 QX'D'\epsilon}_{o_P(\frac{1}{\sqrt{N}})} = & -\lambda_2 (\pi'Z'Z\pi)^{-1}\eta'D'\epsilon -\lambda_2 (\pi'Z'Z\pi)^{-1}\pi'ZD'\epsilon + o_P(\frac{1}{N})\\
    = &  - \underbrace{\frac{Q_0}{N} \eta'\lambda_2D'\epsilon}_{(e)} -\underbrace{\lambda_2 (\pi'Z'Z\pi)^{-1}\pi'Z'D'\epsilon}_{E[\lambda_2 (\pi'Z'Z\pi)^{-1}\pi'Z'D'\epsilon] = 0} + o_P(\frac{1}{N})\\
\end{align*}
Note that $R_3 = (d) - (e)$.

\section{Proof for approximate bias of UIJIVE1 is asymptotically vanishing} \label{appendix UIJIVE1 approximate bias asymptotically vanishing}

\begin{align*}
    tr(C) - \mathcal{L} - 1 = \sum_{i=1}^{N} \frac{\frac{L_1+1}{N}}{1 - \tilde{D}_i + \frac{L_1+1}{N}} - L_1 - 1 =&  \sum_{i=1}^{N} \frac{L_1+1}{N - N\tilde{D}_i + L_1+1} - L_1 - 1 \\
    =& \sum_{i=1}^{N} \{\frac{L_1+1}{N - N\tilde{D}_i + L_1+1} - \frac{L_1 + 1}{N}\}\\
    =& \sum_{i=1}^{N} \frac{N\tilde{D}_i(L_1+1) - (L_1+1)^2}{(N - N\tilde{D}_i + L_1+1)N}
\end{align*}


\begin{claim}
    $tr(C) - \mathcal{L} - 1 \to 0$. 
\end{claim}
\begin{proof}
    \begin{align*}
        \lVert \sum_{i=1}^{N} \frac{\frac{L_1+1}{N}}{1 - \tilde{D}_i + \frac{L_1+1}{N}} - L_1 - 1 \rVert
        \leq&  \lVert \sum_{i=1}^{N} \frac{N\tilde{D}_i(L_1+1)}{(N - N\tilde{D}_i + L_1+1)N} \rVert + \lVert \sum_{i=1}^{N} \frac{(L_1+1)^2}{(N - N\tilde{D}_i + L_1+1)N} \rVert\\
        \leq & 
        \lVert \sum_{i=1}^{N} \frac{\tilde{D}_i(L_1+1)}{(mN + L_1+1)} \rVert + \lVert \sum_{i=1}^{N} \frac{(L_1+1)^2}{(mN + L_1+1)N} \rVert \quad \text{(BA)}\\
        =& \frac{K_1(L_1+1)}{mN + L_1 +1} + \frac{(L_1+1)^2}{mN + L_1+1} \to 0
    \end{align*}
\end{proof}

\section{Generalizing approximate bias to other classes of estimators} \label{sec other classes of estimators}
\subsection{k-class estimators} \label{approximately unbiased k class estimator}
Classical k-class estimators takes the form of $(X'C'X)^{-1}(X'C'y)$ and its $C$ is an affine combination of $C_{OLS}$ ($=I$) and $C_{TSLS}$ ($=P_Z$). Its $C$ matrix satisfies the $CZ = Z$ property.
\begin{equation*}
	k C_{TSLS} Z  + (1-k) C_{OLS} Z =  kZ + (1-k) Z = Z \quad \text{where} \quad \alpha \in \mathbb{R}.
\end{equation*}
Therefore, corollary \ref{corollary for approximate bias} applies to all k-class estimators. I set $k = \frac{N-L-1}{N-K}$ such that the approximate bias of the k-class estimator is zero as in Eq(\ref{approximaely unbiased k-class estimator equation}). The resulting estimator is termed an Approximately Unbiased k-class estimator (or AUK in short) and AUK's $k$ converges at a rate of $O(\frac{1}{N^2})$ to that of the Nagar estimator. In contrast, the Nagar estimator's k converges to that of TSLS ($k=1$) at a rate of $O(\frac{1}{N})$. 
\begin{align}
	tr(k C_{TSLS} + (1-k) C_{OLS}) - L - 1 &= 0 \label{approximaely unbiased k-class estimator equation} \\
	kK - (1-k)N - L - 1 &= 0 \nonumber\\
	k = 1+ \frac{K-L-1}{N - K} &= \underbrace{ 1 + \frac{K-L-1}{N}}_{\text{Nagar estimator's $k$}} + O(\frac{1}{N^2}) \nonumber
\end{align}

\subsection{$\lambda_1$-class estimator}\label{Approximately Unbiased lambda class}
$\lambda_1$-class estimator bridges between JIVE1 and TSLS, both of which have the $CZ = Z$ property. To maintain this property, the $C$ matrix of $\lambda_1$-class estimator is designed to be $(I-\lambda_1D)^{-1}(P_Z - \lambda_1D)$ such that when $\lambda_1 = 0$, the estimator is TSLS; when $\lambda_1 = 1$, the estimator is JIVE1. Evaluating the approximate bias of $\lambda_1$-class estimator, we get
\begin{equation*}
    (I-\lambda_1D)^{-1}(P_Z - \lambda_1D)Z = (I-\lambda_1D)^{-1}(Z - \lambda_1DZ) = (I-\lambda_1D)^{-1}(I-\lambda_1D)Z = Z.
\end{equation*}
By Corollary \ref{corollary for approximate bias}, the approximate bias of $\lambda$-class estimator is proportional to
\begin{equation}\label{approximate bias of lambda class}
    (1-\lambda_1)\sum_{i=1}^{N} \frac{D_i}{1-\lambda_1 D_i} -L-1 
\end{equation}

Under Assumption BA, setting $\lambda_1 = \frac{K-L-1}{K}$ makes the approximate bias asymptotically vanish.

\subsection{$\lambda_2$-class estimator}
The relationship between $\lambda_1$-class estimator and $\lambda_2$-class estimator is analogous to the relationship between JIVE1 and JIVE2. $\lambda_2$-class estimator removes the row-wise division of $\lambda_1$-class estimator. Hence, the $C$ matrix of $\lambda_2$-class estimator is designed to be $P_Z - \lambda_2 D$.

By corollary \ref{corollary for approximate bias}, the approximate bias for $\lambda_2$-class estimator is $\lambda_2K -L - 1$, and hence, the approximately unbiased $\lambda_2$-class estimator has its $\lambda_2 = \frac{K-L-1}{K}$ and I call this estimator TSJI2. 

\subsection{$\omega_1$-class estimator: UOJIVE1}\label{Approximately Unbiased omega class}
By Corollary \ref{corollary for approximate bias}, the approximate bias of $\omega_1$-class estimator is proportional to 
\begin{equation} \label{approximate bias of omega class}
	\sum_{i=1}^{N}\frac{\omega_1}{1-D_i+\omega_1} - L -1.
\end{equation}
By a similar proof to Appendix \ref{appendix UIJIVE1 approximate bias asymptotically vanishing}, we set $\omega_1 = \frac{L+1}{N}$.

\subsection{$\omega_2$-class estimator: UOJIVE2}
Similarly, we $\omega_2 = \frac{L+1}{N}$.

\section{Asymptotic proofs} \label{appendix asymptotic proofs}

\subsection{Proof for Lemma~\ref{prop UOJIVE1 consistency A}}
\begin{align}
    &\frac{1}{N}X'P_Z(I-D-\omega I)^{-1}X = \frac{1}{N}X'Z(Z'Z)^{-1}\sum_{i=1}^{N}\frac{Z_i'X_i}{1-D_i+\omega} \label{E1}\\
    &\frac{1}{N}X'D(I-D-\omega I)^{-1}X = \frac{1}{N}\sum_{i=1}^{N}\frac{D_iX_i'X_i}{1-D_i+\omega} \label{E2} \\
    &\frac{1}{N}X'\omega I(I-D-\omega I)^{-1}X = \frac{1}{N}\sum_{i=1}^{N}\frac{\omega X_i'X_i}{1-D_i+\omega} \label{E3}
\end{align}

Consider expression (\ref{E1}), 
\begin{align*}
	\lVert
	\frac{1}{N}\sum_{i=1}^{N}\frac{Z_i'X_i}{1-D_i+\omega} - \frac{1}{N}Z'X \rVert \leq & \frac{1}{N} \sum_{i=1}^{N} \lVert \frac{(D_i-\omega)Z_i'X_i}{1-D_i+\omega} \rVert\\
	\leq & \frac{1}{N} \frac{K}{m} \max_i \lVert Z_i'X_i \rVert + \frac{1}{N} \frac{L+1}{m} \max_i \lVert Z_i'X_i \rVert
\end{align*}
Assumption \ref{finite XX} and Lemma~\ref{lemma moment to max order} jointly imply that $\lVert
\frac{1}{N}\sum_{i=1}^{N}\frac{Z_i'X_i}{1-D_i+\omega} - \frac{1}{N}Z'X \rVert$  is bounded above by the sum of of two $o_P(1)$ terms' norms. Therefore, expression (\ref{E1}) converges in probability to $H$.

Consider expression (\ref{E2}), 
\begin{equation*}
	\lVert \frac{1}{N}\sum_{i=1}^{N}\frac{D_iX_i'X_i}{1-D_i+\omega} \rVert \leq \frac{1}{N}\sum_{i=1}^{N} \lVert \frac{D_iX_i'X_i}{1-D_i+\omega} \rVert \leq \frac{K}{Nm} \max_i \lVert X_i'X_i \rVert.
\end{equation*}
Assumption \ref{finite XX} and Lemma~\ref{lemma moment to max order} jointly imply that expression (\ref{E2}) converges in probability to 0.

Consider expression (\ref{E3}), 
\begin{equation*}
	\lVert \frac{1}{N}\sum_{i=1}^{N}\frac{\omega X_i'X_i}{1-D_i+\omega} \rVert \leq \frac{\lim_{N\to\infty}\omega}{m} E[\lVert X_i'X_i \rVert] \to 0
\end{equation*}
because $\omega = O(\frac{1}{N})$.

\subsection{Proof for Lemma~\ref{prop UOJIVE1 consistency B}}

\begin{align}
		&\frac{1}{N} X'P_Z(I-D+\omega I)^{-1}\epsilon= \frac{1}{N} X'Z(Z'Z)^{-1}\sum_{i=1}^{N}\frac{Z_i'\epsilon_i}{1-D_i+\omega}  \label{EEE1}\\
		&\frac{1}{N} X'D(I-D+\omega I)^{-1}\epsilon= \frac{1}{N} \sum_{i=1}^{N}\frac{D_iX_i'\epsilon_i}{1-D_i+\omega}  \label{EEE2} \\
		& \frac{1}{N}X'\omega I(I-D-\omega I)^{-1}\epsilon = \frac{1}{N}\sum_{i=1}^{N}\frac{\omega X_i'\epsilon_i}{1-D_i+\omega}  \label{EEE3}
	\end{align}

Consider $\frac{1}{N}\sum_{i=1}^{N}\frac{Z_i'\epsilon_i}{1-D_i+\omega}$ in expression (\ref{EEE1}),
\begin{align*}
	\lVert\frac{1}{N}\sum_{i=1}^{N}\frac{Z_i'\epsilon_i}{1-D_i+\omega} -\frac{1}{N}Z'\epsilon\rVert \leq& \frac{1}{N}\sum_{i=1}^{N} \lVert \frac{D_iZ_i'\epsilon_i - \omega Z_i'\epsilon_i}{1-D_i+\omega} \rVert \\
	\leq& \frac{1}{Nm}\sum_{i=1}^{N} \lVert D_iZ_i'\epsilon_i - \omega Z_i'\epsilon_i\rVert \quad \text{(BA)}\\
	\leq & \frac{K}{Nm} \max_i \lVert Z_i'\epsilon_i \rVert + \frac{N\omega}{Nm}\max_i \lVert Z_i'\epsilon_i \rVert
\end{align*}
Both terms converge in probability to zero under Assumption \ref{finite Xepsilon lower}. Note that $N\omega = O(1)$. Therefore, $\frac{N\omega}{Nm}\max_i \lVert Z_i'\epsilon_i \rVert = O(1)o_P(1) = o_P(1)$.

Consider expression (\ref{EEE2}),
\begin{equation*}
	\lVert\frac{1}{N} \sum_{i=1}^{N}\frac{D_iX_i'\epsilon_i}{1-D_i+\omega}\rVert \leq \frac{1}{Nm}\sum_{i=1}^{N} \lVert D_iX_i'\epsilon_i \rVert\leq \frac{K}{Nm} \max_i\lVert X_i'\epsilon_i \rVert
\end{equation*}
The last terms converge in probability to zero under Assumption \ref{finite Xepsilon lower}.

Consider expression (\ref{EEE3}),
\begin{equation*}
	 \lVert \frac{1}{N}\sum_{i=1}^{N}\frac{\omega X_i'\epsilon_i}{1-D_i+\omega}\rVert \leq \frac{1}{N}\sum_{i=1}^{N} \lVert \frac{\omega X_i'\epsilon_i}{1-D_i+\omega} \rVert \leq \frac{\omega}{Nm} \max_i\lVert X_i'\epsilon_i \rVert
\end{equation*}
The last term converges in probability to zero under Assumption \ref{finite Xepsilon lower}.

\subsection{Proof for Lemma~\ref{prop UOJIVE2 consistency A}}
\begin{align*}
    &\frac{1}{N}X'P_ZX \overset{p}{\to} H \\
    &\frac{1}{N}X'DX = \frac{1}{N}\sum_{i=1}^{N}D_iX_i'X_i \leq \frac{K}{N} \max_i \lVert X_i'X_i \rVert \overset{p}{\to} 0 \\
    &\frac{1}{N}\omega X'X = \frac{1}{N}\sum_{i=1}^{N} \omega X_i'X_i \overset{p}{\to} 0
\end{align*}

\subsection{Proof for Lemma~\ref{prop UOJIVE2 consistency B}}
\begin{align*}
    &\frac{1}{N} X'P_Z\epsilon= \frac{1}{N} X'Z(Z'Z)^{-1}Z'\epsilon \overset{p}{\to} 0 \\
    &\frac{1}{N} X'D\epsilon= \frac{1}{N} \sum_{i=1}^{N} D_iX_i'\epsilon_i \leq \frac{K}{N} \max_i \lVert X_i'\epsilon_i \rVert \overset{p}{\to} 0 \\
    & \frac{1}{N}\omega X'\epsilon = \frac{1}{N}\sum_{i=1}^{N} \omega X_i'\epsilon_i \overset{p}{\to} 0
\end{align*}

\subsection{Proof for Lemma~\ref{prop asymptotic variance UOJIVE1 B}}
The proof is similar to proof for Lemma~\ref{prop UOJIVE1 consistency B}.
I first show that $\frac{1}{\sqrt{N}}\sum_{i=1}^{N}\frac{Z_i'\epsilon_i}{1-D_i+\omega} - \frac{1}{\sqrt{N}}Z'\epsilon \overset{p}{\to} 0$.
	
	\begin{align*}
		\lVert\frac{1}{\sqrt{N}}\sum_{i=1}^{N}\frac{Z_i'\epsilon_i}{1-D_i+\omega} - \frac{1}{\sqrt{N}}\sum_{i=1}^{N} Z_i'\epsilon_i \rVert= & \frac{1}{\sqrt{N}}\sum_{i=1}^{N} \lVert\frac{(D_i - \omega)Z_i'\epsilon_i}{1-D_i+\omega}\rVert\\
		= & \frac{1}{\sqrt{N}}\sum_{i=1}^{N} \lVert\frac{D_iZ_i'\epsilon_i}{1-D_i+\omega}\rVert + \frac{1}{\sqrt{N}}\sum_{i=1}^{N} \lVert\frac{ \omega Z_i'\epsilon_i}{1-D_i+\omega}\rVert\\
		\leq & \frac{1}{\sqrt{N}}\sum_{i=1}^{N} \lVert\frac{D_iZ_i'\epsilon_i}{m+\omega}\rVert + \frac{1}{\sqrt{N}}\sum_{i=1}^{N} \lVert\frac{ \omega Z_i'\epsilon_i}{m+\omega}\rVert\\
		\leq & \frac{K}{(m+\omega)\sqrt{N}}\max_i\lVert Z_i'\epsilon_i\rVert + \frac{\omega}{(m+\omega)\sqrt{N}}\sum_{i=1}^{N} \lVert Z_i'\epsilon_i\rVert
	\end{align*}
The first term converges to 0 in probability under Assumption \ref{finite Xepsilon}. The second term converges to 0 in probability $\omega = O(\frac{1}{N})$. Therefore, $\frac{1}{\sqrt{N}}\sum_{i=1}^{N}\frac{Z_i'\epsilon_i}{1-D_i+\omega}  \overset{d}{\to} N(0,\sigma^2 \Sigma_{ZZ})$ and $\frac{1}{\sqrt{N}}X'P_Z (I -  D + \omega I)^{-1}\epsilon \overset{d}{\to} N(0,\sigma_\epsilon^2 H)$.

The other two terms converge to 0 in probability under Assumption BA and \ref{finite Xepsilon}
\begin{align*}
	\lVert \frac{1}{\sqrt{N}}X'D (I -  D + \omega I)^{-1}\epsilon\rVert \leq \frac{1}{\sqrt{N}} \sum_{i=1}^{N} \lVert \frac{D_i X_i'\epsilon_i}{1-D_i +\omega} \rVert &\leq \frac{1}{\sqrt{N}} \frac{K}{m+\omega}\max_i\lVert X_i'\epsilon_i\rVert \overset{p}{\to} 0,\\
	\lVert \frac{1}{\sqrt{N}}X'\omega I (I -  D + \omega I)^{-1}\epsilon\rVert \leq \frac{1}{\sqrt{N}} \sum_{i=1}^{N} \lVert \frac{\omega X_i'\epsilon_i}{1-D_i+\omega} \rVert &\leq  \frac{N\omega}{m+\omega}\frac{1}{\sqrt{N}}\max_i\lVert X_i'\epsilon_i\rVert \\
	&= \frac{L+1}{m+o(1/N)}o_P(1) = o_P(1).
\end{align*}

\subsection{Proof for Lemma~\ref{prop asymptotic variance UOJIVE2 B}}
\begin{align*}
    \frac{1}{\sqrt{N}}{X}'P_{{Z}}'{\epsilon} =&  \frac{1}{\sqrt{N}}{X}'{Z} ({Z}'{Z})^{-1}{Z}'{\epsilon} \overset{d}{\to} N(0,\sigma_{{\epsilon}}^2\Sigma_{{X}'{Z}}\Sigma_{{Z}'{Z}}^{-1} \Sigma_{{Z}'{X}}) \\
    \frac{1}{\sqrt{N}}{X}' {D}'{\epsilon} =& \frac{1}{\sqrt{N}} \sum_{i=1}^{N} {D}_i{X}_i'{\epsilon}_i \overset{p}{\to} 0 \\
    \frac{L+1}{\sqrt{N}}\frac{1}{N}{X}'{\epsilon} =& O(\frac{1}{\sqrt{N}})O_P(1) = O_P(\frac{1}{\sqrt{N}}) = o_P(1)
\end{align*}

\section{Example for the simulation setup with outliers} \label{sec matrix with outlier}

\begin{table}[ht]
\centering
\begin{tabular}{|c|c|c|c|c|l|}
\hline
\textbf{Column 1} & \textbf{Column 2} & \textbf{Column 3} & \textbf{Column 4} & \textbf{Column 5} & \textbf{Description} \\ \hline
\( (N-1)^{1/3} \) & 0 & 0 & 0 & 0 & Outlier \\ \hline
1 & 0 & 0 & 0 & 0 & Group 2 (Row 2) \\ \hline
0 & 1 & 0 & 0 & 0 & Group 3 (Row 3) \\ \hline
0 & 0 & 1 & 0 & 0 & Group 4 (Row 4) \\ \hline
0 & 0 & 0 & 1 & 0 & Group 5 (Row 5) \\ \hline
0 & 0 & 0 & 0 & 1 & Group 6 (Row 6) \\ \hline
0 & 0 & 0 & 0 & 0 & Group 1 (Rows 7-11) \\ \hline
0 & 0 & 0 & 0 & 0 & Group 1 (continued) \\ \hline
0 & 0 & 0 & 0 & 0 & Group 1 (continued) \\ \hline
0 & 0 & 0 & 0 & 0 & Group 1 (continued) \\ \hline
0 & 0 & 0 & 0 & 0 & Group 1 (continued) \\ \hline
1 & 0 & 0 & 0 & 0 & Group 2 (Row 12) \\ \hline
0 & 1 & 0 & 0 & 0 & Group 3 (Row 13) \\ \hline
0 & 0 & 1 & 0 & 0 & Group 4 (Row 14) \\ \hline
0 & 0 & 0 & 1 & 0 & Group 5 (Row 15) \\ \hline
0 & 0 & 0 & 0 & 1 & Group 6 (Row 16) \\ \hline
0 & 0 & 0 & 0 & 0 & Group 1 (Rows 17-21) \\ \hline
0 & 0 & 0 & 0 & 0 & Group 1 (continued) \\ \hline
\multicolumn{6}{|c|}{\dots (Pattern continues until 101 rows)} \\ \hline
\end{tabular}
\caption{Matrix Structure for Simulation with Outliers for $N = 101$, the outlier is supposed to belong to Group 2, but has its value contaminated (multiplied by $(N-1)^{1/3}$).}
\label{table matrix with outlier}
\end{table}

\end{appendices}
\newpage
\bibliographystyle{plainnat}
\bibliography{reference}

\end{document}